\newcount\Comments  
\documentclass{article}
\usepackage[margin=1in]{geometry}
\usepackage{amsmath,amsfonts,amssymb,amsthm,bm}
\usepackage{booktabs}
\usepackage{enumitem}
\usepackage{comment}
\usepackage{color}
\usepackage{graphicx}
\usepackage{pgf,pgfplots}
\usepackage{natbib}

\newcommand{\E}{\mathbb{E}}
\newcommand{\f}{\mathbf{f}}
\newcommand{\G}{\mathbf{G}}
\newcommand{\p}{\mathbf{p}}
\newcommand{\Q}{\mathbf{Q}}
\renewcommand{\r}{\mathbf{r}}
\newcommand{\R}{\mathbf{R}}
\newcommand{\s}{\mathbf{s}}
\renewcommand{\S}{\mathbf{S}}
\renewcommand{\P}{\mathbf{P}}

\renewcommand{\tilde}{\widetilde}

\newtheorem{theorem}{Theorem}[section]
\newtheorem{conjecture}[theorem]{Conjecture}
\newtheorem{definition}[theorem]{Definition}
\newtheorem{lemma}[theorem]{Lemma}
\newtheorem{proposition}[theorem]{Proposition}
\newtheorem{example}[theorem]{Example}
\newtheorem{corollary}[theorem]{Corollary}
\newtheorem{assumption}[theorem]{Assumption}

\newcommand{\ignore}[1]{}
\newcommand{\kibitz}[2]{\ifnum\Comments=1\textcolor{#1}{#2}\fi}

\newcommand{\sherry}[1]{\kibitz{red}{\noindent[Sherry: #1]}}
\newcommand{\yc}[1]{\kibitz{magenta}{\noindent[YC: #1]}}

\newcommand{\SR}[1]{{\kibitz{red}{[Sherry: #1]}}}
\newcommand{\fang}[1]{{\kibitz{blue}{\noindent[Fang: #1]}}}

\begin{document}

\title{The Limits of Multi-task Peer Prediction}

\author{
Shuran Zheng \\
Harvard University\\
{\normalsize shuran\_zheng@seas.harvard.edu}\\
\and
Fang-Yi Yu \\
Harvard University\\
{\normalsize fangyiyu@seas.harvard.edu}
\and
Yiling Chen \\
Harvard University\\
{\normalsize yiling@seas.harvard.edu}\\
}
\date{}


\maketitle
\begin{abstract}

Recent advances in multi-task peer prediction have greatly expanded our knowledge about the power of multi-task peer prediction mechanisms. Various mechanisms have been proposed in different settings to elicit different types of information. But we still lack understanding about when desirable mechanisms will exist for a multi-task peer prediction problem. In this work, we study the elicitability of multi-task peer prediction problems. 
We consider a designer who has certain knowledge about the underlying information structure and wants to elicit certain information from a group of participants. Our goal is to infer the possibility of having a desirable mechanism based on the primitives of the problem.

Our contribution is twofold. First, we provide a characterization of the elicitable multi-task peer prediction problems, assuming that the designer only uses \emph{scoring mechanisms}. Scoring mechanisms are the mechanisms that reward participants' reports for different tasks separately. The characterization uses a geometric approach based on the \emph{power diagram} characterization~\citep{lambert2009eliciting,frongillo2017geometric} in the single-task setting. For general mechanisms, we also give a necessary condition for a multi-task problem to be elicitable. 

Second, we consider the case when the designer aims to elicit some properties that are linear in the participant's posterior about the state of the world. We first show that in some cases, the designer basically can only elicit the posterior itself. We then look into the case when the designer aims to elicit the participants' posteriors. We give a necessary condition for the posterior to be elicitable. This condition implies that the mechanisms proposed by \citet{kong2018water} are already the best we can hope for in their setting, in the sense that their mechanisms can solve any problem instance that can possibly be elicitable.
\end{abstract}

\section{Introduction}
\sherry{opening paragraph that highlights the importance of peer prediction?}

\yc{Added the first paragraph and changed the second.}

\yc{I changed underlying distribution to underlying information structure. Not sure which one is better. In either case, it's unclear which distribution we are talking about.}

Peer prediction refers to a collection of incentive mechanisms
~\citep{MRZ05,prelec2004bayesian,jurca2008incentives,radanovic2013robust,radanovic2014incentives,witkowski2012peer,dasgupta2013crowdsourced,shnayder2016informed,schoenebeck2020learning,kong2019information,kong2018water,liu2020surrogate} 
that have been designed for the challenging setting where truthful information elicitation about some tasks is desired but the designer has no access to the ground truth (i.e. event outcomes) for incentive alignment. This setting is fundamental to many information elicitation applications such as peer grading, surveys, product reviews, and forecasting for long-term events. 

Recent advances in peer prediction have progressed from single-task peer prediction~\citep{MRZ05,prelec2004bayesian,jurca2008incentives,radanovic2013robust,radanovic2014incentives,witkowski2012peer}, where an agent's reward on a task is solely determined by how his report on the task relates to the reports made by peer agents on the same task, to multi-task peer prediction~\citep{dasgupta2013crowdsourced,shnayder2016informed,schoenebeck2020learning,kong2019information,kong2018water,liu2020surrogate}, where reports made by peer agents on other tasks can also be used in determining the agent's reward on the task. Multi-task peer prediction mechanisms can often achieve stronger incentive guarantees or require fewer assumptions on the underlying information structure than single-task peer prediction mechanisms, thanks to the additional cross-task information. While the quest for better peer prediction mechanisms is bound to continue, we attempt to understand the limits for designing multi-task peer prediction mechanisms in this work: When is it possible to design a desirable multi-task peer prediction mechanism?


To answer this question, the first thing we may need to understand is: what are the factors that will influence the elicitability of a problem?
Our first observation is that the existing multi-task mechanisms (as well as single-task mechanisms) all rely on certain knowledge about the underlying information structure or various assumptions about it. 
For example, \citet{shnayder2016informed} required the designer to partially know the correlation between the participants' signals; \citet{kong2018water} and \citet{liu2020surrogate} assumed that the participants' signals are mutually independent conditioned on the unknown ground truth. 
In addition to the various assumptions about the underlying distribution, we have also seen mechanisms that utilize the structure of the reported information. For example, \citet{kong2018water} proposed a mechanism that rewards the \emph{point-wise mutual information} between the participants' reports, which can only be computed when the reports are the participants' posteriors about the state. 
\sherry{Do you agree with this statement? A lot of works still consider minimal setting?}\yc{Don't understand why these papers are non-minimal settings. At least the second paper only elicits predictions and no additional information.} \sherry{Is a mechanism that elicits prediction a minimal mechanism?}\yc{A minimal mechanism elicits the information that it intends to elicit. I think it can be predictions or signals. Basically minimal mechanisms do not elicit other information.}\yc{Hmm, I see. The Frongilo and Witkowiski's paper indeed defines minimal mechanisms as mechanisms that only elicit signals. But I thought that was only because the peer prediction problem was original defined with the goal of eliciting signals. I couldn't find any reference to minimal mechanisms now. It's a term spiritually borrowed from mechanism design. I wonder whether we want to avoid using the term minimal in this paper.} \sherry{Yes, I think I've removed the term minimal when it may lead to ambiguity.}

 Our problem becomes more clear. Suppose there is a designer who has certain knowledge about the underlying information structure and she wants to elicit certain information from a group of participants, can we infer the possibility of having a desirable mechanism based on the primitives of the problem before trying to search for mechanisms? The answer to this question may also shed light on the design of new mechanisms. For certain information we want to elicit, what do we have to know about the underlying information structure? Based on our knowledge about the information structure, what kind of information can we possibly elicit?

This problem has been studied in single-task peer prediction if we consider 
mechanisms that only ask the participants for their signals. \citet{frongillo2017geometric} used a geometric perspective to prove that single-task peer prediction mechanisms that achieve strict truthfulness are equivalent to \emph{power diagrams}. Their result gives a necessary and sufficient condition for a designer with certain knowledge about the participants' posterior beliefs to be able to design a strictly truthful mechanism: the designer should be able to divide a participant' possible posteriors after seeing different signal realizations into different regions, and moreover, these regions must take a particular shape, that of a \emph{power diagram}. \yc{This paragraph still uses the term minimal.}\sherry{addressed}

But for multi-task peer prediction, little is known about the exact condition for the existence of strictly truthful mechanisms. 
As we will show by an example (at the beginning of Section~\ref{sec:mult}), it is possible for the designer to exploit the similarity between the tasks and elicit information that is not elicitable in the single-task setting, assuming that the designer has the same knowledge about the distribution. 
A natural thought one may have is to view a multi-task problem as a single-task problem in which a participant's report is a combination of reports for multiple tasks. The problem of directly converting a multi-task problem into a single task problem is that a report will have exponentially many possible values and the condition given by \citet{frongillo2017geometric} will involve power diagrams in  dimension that grows exponentially in the number of tasks, which may not lead to meaningful results. Even for a constant number of tasks, directly applying their method does not give an easily interpretable characterization.


In this work, we study the elicitability in multi-task peer prediction. For the definition of elicitability, we consider the most basic incentive guarantee that truth-telling is a strict Bayesian Nash equilibrium (BNE)\fang{add BNE}. We say that a peer prediction problem is elicitable if there exists a mechanism that guarantees that truth-telling is a strict BNE for any possible underlying information structure. We consider a peer prediction problem to have two primitives. The first one is the designer's knowledge or assumption about the underlying information structure. The knowledge may have different forms. In this work, we model the designer's knowledge by a set of possible underlying information structures. The designer knows that the underlying information structure must lie in this set, but she does not know which one is the true one. The second primitive is the information that the designer asks each participant to report. In this work, we consider the most general information which can be a function of the information structure and the participant's signal. \yc{I see why there is confusion with minimal mechanisms. If the principal intends to get the signal of an agent and asks for agents to only report her signal, the mechanism is minimal. If in addition to reporting the signal, the agent is asked to report other information (e.g. in BTS and robust BTS), the mechanism is not minimal. If the principal is interested in getting the mean prediction of an event (or some other function of the distribution), and asks agents to report just that, it's still a minimal mechanism.} \sherry{But why can't the designer just include all the information that is required into the information that she is interested? In BTS, why not just define a signal to be both the signal and the prediction? Then every mechanism will be a non-minimal mechanism?}\yc{Well, if we define signal and prediction together as signal, then that puts a huge constraint on what signal distribution is possible. Also, the point about minimal mechanisms is that whether the mechanisms only ask for what the mechanism designer cares about and no more.} 

Our contribution is twofold. First, we give a characterization of the elicitable multi-task peer prediction problems, assuming that the designer only uses \emph{scoring mechanisms}. Scoring mechanisms are the mechanisms that reward participants' reports for different tasks separately. To our knowledge, all the existing mechanisms that achieve strict truthfulness are scoring mechanisms.\footnote{The only mechanism that we know does not belong to scoring mechanisms is the Determinant-based Mutual Information mechanism proposed by~\citet{kong2020dominantly}. The mechanism is not strictly truthful because it cannot distinguish permutation strategies from truth-telling.} 
We show that a multi-task problem is elicitable if and only if the following two conditions are satisfied: (1) the designer should be able to separate a participant's possible posteriors after seeing different signal realizations using a power diagram, \emph{for any given marginal distribution of other participants' truthful reports}; (2) for different marginal distributions of other participants' truthful reports, the parameters of the power diagrams should be an affine function of 
the marginal distribution of other participants' truthful reports.\fang{add tensor power here}\SR{I feel that tensor power is too difficult to digest. Maybe we don't have to be that formal here. We didn't indicate whether this is the distribution for one task or many tasks anyway.}\fang{then we can add informally, and a pointer to the real statement}
For general mechanisms, we give a necessary condition for a multi-task problem to be elicitable. The necessary condition basically says that, given a joint distribution of the participants' reports without naming a participant $i$'s report, the designer should at least be able to label participant $i$'s report based on the distribution.

Second, we consider the case when the designer aims to elicit some properties that are linear in the participant's posterior about the state of the world. We first show that in some cases, the designer basically can only elicit the posterior itself. More specifically, we apply our characterization to the case when there are two participants with signals independent conditioning on the state. If the designer only uses scoring mechanisms, then the only elicitable linear properties of the posterior are the ones that are equivalent to the posterior, assuming that the designer is uncertain about the underlying distribution. We then look into the case when the designer elicits the participants' posteriors. We give a necessary condition for the posterior to be elicitable. This condition implies that the mechanisms proposed by \citet{kong2018water} are already the best we can hope for in their setting, in the sense that their mechanisms can work for any problem instance that can possibly be elicitable.

\subsection{Related work}
The elicitability of peer prediction problems has not received a lot of attention.
For the single-task peer prediction, \citet{frongillo2017geometric} use a geometric approach to study necessary and sufficient conditions for the existence of strictly truthful peer prediction mechanisms.  However, their characterizations are for single-task mechanisms that only collect agents' signals, but many single-task mechanisms elicit information other than or in addition to agents' signals or require relatively strong assumptions on the underlying information structure.  
\citet{zhang2014elicitability} also consider the existence of strictly truthful mechanisms.  They show stochastic relevance is a necessary condition even for general mechanisms.

Although the problem of elicitability has not been extensively investigated in the peer prediction literature. There is a vast literature on \emph{property elicitation}. We are not able to review all of the works in this area but point the readers to  \citep{frongillo2013eliciting} and the references therein. In property elicitation, the designer asks an agent to report a property of a probability distribution. The designer is able to observe a sample drawn from the distribution and then decide the payment based on the report and the sample. To a certain extent, a peer prediction problem can be viewed as a property elicitation problem in which the other participants' reports is a sample. But the problem is that it may not always be possible to represent the information that the designer wants to elicit as a property of the distribution of other participants' reports.

Finally, we review the existing literature on peer prediction, in both multi-task setting and single-task setting. Also see~\citep{faltings2017game} for a survey of additional results.
\paragraph{Multi-task setting}
The multi-task peer prediction problem was first independently introduced and studied by \citet{dasgupta2013crowdsourced} and \citet{Witkowski2013LearningTP}.  Agents are assigned a batch of a priori similar tasks which require each agents' private information to be a binary signal. Later works extend the setting to multiple-choice questions and design mechanisms that achieve various truthfulness guarantees (dominant truthful, informed truthful)~\citep{schoenebeck2020learning,kong2020dominantly, kong2019information,shnayder2016informed,dasgupta2013crowdsourced}. But none of them is strictly truthful on general distributions, because agents can always relabel their signals.   \citet{liu2020surrogate} design an approximated dominant truthful mechanism (also approximated strictly truthful) that uses surrogate loss functions as tools to correct for the mistakes in agents' reports. \citet{kong2018water} study the related goal for forecast elicitation.  All the above mechanisms~\citep{dasgupta2013crowdsourced,shnayder2016informed,schoenebeck2020learning,kong2019information,kong2018water,liu2020surrogate} are scoring mechanisms (Definition~\ref{def:scoring}) except the DMI mechanism by ~\citet{kong2020dominantly} and the VMI mechanisms by \citet{kong2021counting}.\fang{(notice) another yuqing's work}

\paragraph{Single-task setting}  \citet{MRZ05} introduced the original peer prediction mechanism, which is the first mechanism that has truth-telling as a strict Bayesian Nash equilibrium and does not need verification. However, their mechanism requires the full knowledge of the common prior. \citet{prelec2004bayesian} relaxes the full knowledge assumption and designs the first detail-free peer prediction mechanism---Bayesian truth serum (BTS). BTS requires that all agents' signals are symmetric and conditional independent given a latent state. Several other works study the same single-task setting as BTS and devise mechanisms that work on more general underlying information structures~\citep{jurca2008incentives,radanovic2013robust,radanovic2014incentives,witkowski2012peer,kong2016equilibrium, schoenebeck2020two}. 



\section{Problem Description}
Consider a designer who wants to elicit certain information about the state of the world $\omega\in \Omega$ from a group of participants. 
There are $n$ participants who receive private signals $s_1, \dots, s_n$ respectively with $s_1 \in \mathcal{S}_1, \dots, s_n \in \mathcal{S}_n$. We use $S_i$ to denote the random variable for participant $i$'s signal. The state of the world and the signals follow an unknown underlying distribution $\mu(\omega, s_1, \dots, s_n)$. The designer does not know the true underlying distribution $\mu$, but she may have some information about the structure of the distribution, which allows her to restrict $\mu$ to a set  $M \subseteq \Delta(\Omega \times \mathcal{S}_1 \times \cdots \times \mathcal{S}_n)$. 
We assume that the underlying distribution $\mu(\omega, s_1, \dots, s_n)$ is common knowledge for all the participants. But participant $i$ only observes the realization of his own signal $s_i$ and thus his posterior belief about the state and others' signals will be $\mu(\omega, \s_{-i}|s_i)$, where $\s_{-i}$ denote the signals of the participants other than $i$.

The designer's goal is to elicit certain information from the participants. Participant $i$ will be asked to report a function of his own signal $\r_i = f_{i,\mu}(s_i)$. In this work, we consider functions that are real vectors $\r_i \in \mathbb{R}^L$.  Note that 
the report function can possibly depend on $\mu$. For example, an extensively studied report function is the prediction of the state $\r_i = \mu(\omega|s_i)$; and the well known Bayesian Truth Serum~\citep{prelec2004bayesian} asks for the prediction of other people's signals $\r_i = \mu(\s_{-i}|s_i)$. \yc{What's the space for $r_i$? Can it be a vector? Are we restricting ourselves to minimal peer prediction? BTS is not minimal.} Throughout the work, we use $\r_i = f_{i,\mu}(s_i)$ to represent participant $i$'s \emph{truthful} report. We denote by $\R_i = f_{i,\mu}(S_i)$ the random variable for participant $i$'s truthful report and denote by $\mathcal{R}_i$ the range of the report function $f_{i,\mu}(s_i)$.

In multi-task peer prediction, the designer elicits information for $T>1$ i.i.d. tasks. More specifically, we have $$(\omega^{(1)}, S_1^{(1)}, \dots, S_n^{(1)}), \dots, (\omega^{(T)}, S_1^{(T)}, \dots, S_n^{(T)})\overset{i.i.d}\sim \mu(\omega, s_1, \dots, s_n)$$ where $(\omega^{(t)}, S_1^{(t)}, \dots, S_n^{(t)})$ indicates the state and the signals for task $t$.  The designer elicits the same information from a participant for all $T$ tasks, i.e., $\r_i^{(t)} = f_{i,\mu}(s_i^{(t)})$ for the same function $f_{i,\mu}(\cdot)$ across all the tasks. We denote by $\omega^{(1:T)}$, $\s_i^{(1:T)}$ and $\r_i^{(1:T)}$ the vector of the states, the vector of participant $i$'s signal realizations and the vector of participant $i$'s truthful reports for all $T$ tasks. 

The participants will get paid after reporting the information.
The payment is decided based on the reports across all $T$ tasks. 
\begin{definition}[Multi-task peer prediction mechanism]
A multi-task peer prediction mechanism asks the participants to report their private information $\r_i^{(1:T)}$ for all $T$ tasks. Then the payment to a participant $i$ is decided based on all the reports, denoted by $p_i(\tilde{\r}_1^{(1:T)}, \dots, \tilde{\r}_n^{(1:T)})$ when participant  $i$'s actual report is $\tilde{\r}_i^{(1:T)}$. 
\end{definition}

\fang{alternative:
\begin{definition}[Multi-task peer prediction mechanism]
A multi-task peer prediction mechanism asks the participants to report their private information for all $T$ tasks. Then the payment to a participant $i$ is decided based on all the reports, denoted by $p_i(\tilde{\r}_1^{(1:T)}, \dots, \tilde{\r}_n^{(1:T)})$ when each participant  $i$'s report is $\tilde{\r}_i^{(1:T)}$. 
\end{definition}}
\SR{My concern is that it may not be clear what is the information that the designer asks for.}\fang{that's fine.  I just seldom use the word 'assume' in definitions.}\SR{addressed}

In general, the payment rule for multiple tasks $\p(\tilde{\r}_1^{(1:T)}, \dots, \tilde{\r}_n^{(1:T)})$ can be very complicated. But in practice, we would prefer mechanisms that have succinct payment rules. In this work, we will consider a class of mechanisms that we call the \emph{scoring mechanisms}. 
\begin{definition}[Scoring mechanisms] \label{def:scoring}
	A \emph{scoring mechanism} assigns a payment to each of participant $i$'s report $\tilde{\r}_i^{(t)}$ by comparing it with other participants' report $\tilde{\r}_{-i}$. Formally, a scoring mechanism uses a payment rule that can be represented as follows:
	\begin{align}
		p_i(\tilde{\r}_i^{(1:T)},\ \tilde{\r}_{-i}^{(1:T)}) = \sum_{t=1}^T p_i^{(t)}(\tilde{\r}_i^{(t)}, \tilde{\r}_{-i}^{(1:T)}). 
	\end{align}
\end{definition}
The key feature of scoring mechanisms is that the payment is decided separately for each of participant's reports for different tasks. To our knowledge, all the existing mechanisms that achieve strict truthfulness belong to scoring mechanisms.\fang{except DMI}\SR{We have a footnote in page 3.}

The participants' goal is to maximize their own expected payoff.
 In multiple task peer prediction literature, it is always assumed that an agent's reporting strategy for task $t$ only depends on his signal for that task $s_i^{(t)}$ but not the signals for other tasks $s_i^{(-t)}$. 
So we define a randomized strategy of agent $i$ as follows.
\begin{definition}\label{def:strategy}
	A strategy of agent $i$ for a single task is a mapping $\sigma_i:\mathcal{S}_i \to \Delta(\mathcal{R}_i)$ that maps his observed signal $s_i$ for that task into a distribution of reports, so that when agent $i$ adopts strategy $\sigma_i$, he randomly draw a report $\tilde{\r}_i$ according to $\sigma_i(s_i)$ when the observed signal is $s_i$.
\end{definition}
We denote agent $i$'s strategy for task $t$ by $\sigma_i^{(t)}$.\fang{and assume the mapping $\sigma_i^{(t)}$ are mutually independent for all $i$ and $t$}\fang{not sure why do we use environment to define strategy for signal tasks instead of multi-tasks?}
Then we say a multiple-task mechanism is strictly truthful if truthfully reporting $\r_i^{(t)}$ for all tasks is a strict BNE.
\begin{definition}[Strict truthfulness] \label{def:truthful_multi}
A payment rule $\p(\tilde{\r}^{(1:T)})$ is strictly truthful for a distribution $\mu(\omega, s_1, \dots, s_n)$ if, assuming that the participants know $\mu(\omega, s_1, \dots, s_n)$,
 truthfully reporting $\r_i^{(1:T)}$  is a strict BNE, i.e., for any non-truthful strategy $\sigma_i^{(1:T)}$ with $\sigma_i^{(t)}(s_i^{(t)}) \neq f_{i,\mu}(s_i^{(t)})$ for some $s_i^{(t)}$,
$$
\E_\mu [p_i(\R_i^{(1:T)}, \R_{-i}^{(1:T)})] > \E_\mu [p_i(\sigma_i^{(1:T)}(\S_i^{(1:T)}), \R_{-i}^{(1:T)})].
$$
Here we abuse the notation that $\sigma_i^{(1:T)}(\S_i^{(1:T)})$ represents the vector $(\sigma_i^{(1)}(S_i^{(1)}), \dots, \sigma_i^{(T)}(S_i^{(T)}))$.
\end{definition}

Another commonly used assumption is that the participants use consistent strategies across all tasks. This assumption is usually justified by that the designer can randomly shuffle the tasks so that the participants cannot distinguish the tasks (e.g. see~\citep{shnayder2016informed}). 
\begin{assumption}
The participants use consistent strategies over all the tasks, i.e., $ \sigma_i^{(1)} = \sigma_i^{(2)} = \dots = \sigma_i^{(T)}$ for any participant $i$. 
\end{assumption}
Then we define strict truthfulness under consistent strategies as follows.
\begin{definition}[Strict truthfulness under consistent strategies] \label{def:truthful_multi_consist}
A payment rule $\p(\tilde{\r}^{(1:T)})$ is strictly truthful for a distribution $\mu(\omega, s_1, \dots, s_n)$ if, assuming that the participants know the underlying distribution $\mu(\omega, s_1, \dots, s_n)$,
 truthfully reporting $\r_i^{(1:T)}$  is a strict BNE, i.e., for all  $i$, strategy $\sigma_i$ with $\sigma_i(s_i) \neq f_{i,\mu}(s_i)$ for some $s_i$,
$$
\E_\mu [p_i(\R_i^{(1:T)}, \R_{-i}^{(1:T)})] > \E_\mu [p_i(\sigma_i(\S_i^{(1:T)}), \R_{-i}^{(1:T)})].
$$
Here we abuse the notation that $\sigma_i(\S_i^{(1:T)})$ represents the vector $(\sigma_i(S_i^{(1)}), \dots, \sigma_i(S_i^{(T)}))$.
\end{definition}

The definition of truthfulness is with respect to the single true underlying distribution $\mu$. For a designer who does not know $\mu$ but only knows that $\mu\in M$, we say that $\r$ is elicitable if there exists a mechanism that is strictly truthful for any $\mu \in M$.
\begin{definition}[Elicitability]
\label{def:elicitable}
A multi-task peer prediction problem $\langle \f, M\rangle$ with $\mathbf{f} = \{f_{i,\mu} \}_{(i,\mu)}$
is elicitable if there exists a payment rule $\p(\tilde{\r}^{(1:T)})$ that is strictly truthful for any possible underlying distribution $\mu\in M$. 
\end{definition}
If we assume that the participants use consistent strategies, the the definition of elicitability only requires strict truthfulness under consistent strategies.
\begin{definition}[Elicitability under consistent strategies]
\label{def:elicitable_multi}
A multi-task peer prediction problem $\langle \f, M\rangle$ with $\mathbf{f} = \{f_{i,\mu} \}_{(i,\mu)}$
is elicitable under consistent strategies if there exists a payment rule $\p(\tilde{\r}^{(1:T)})$ that is strictly truthful under consistent strategies for any possible underlying distribution $\mu\in M$.
\end{definition}

In this work, we mainly focus on multi-task peer prediction. Single-task peer prediction can be seen as a special case with $T=1$.

\section{Preliminary}\label{sec:pre}


In this section, we review a few important results from previous works and add some minor findings. First, we give the definition of \emph{power diagrams} and restate the characterization of elicitable problems in the single-task setting. Next, we discuss the correlated agreement mechanism from~\citep{shnayder2016informed} and the mechanism that elicits participants' prediction about the state from~\citep{kong2018water}.

\subsection{Characterization for single-task elicitability} \label{sec:single_elt}

\citet{frongillo2017geometric}  characterized the elicitability of the single-task problem using a geometric approach from the literature on property elicitation~\citep{lambert2008eliciting,lambert2009eliciting}. The basic idea is that the agents' possible posterior beliefs need to fall into the correct regions. The regions are described by \emph{power diagrams}. 

\begin{definition}[\citep{lambert2009eliciting,frongillo2017geometric}] \label{def:power_diag}
	A \emph{power diagram} in dimension $m$ with $K$ cells is a partitioning of $\Delta_m$ into $K$ sets called \emph{cells}, defined by a collection of $K$ $m$-dimensional points $\{\mathbf{v}^k\in \mathbb{R}^m: k\in [K]\}$ called \emph{sites} with associated \emph{weights} $\{w^k\in \mathbb{R}:k\in [K]\}$, given by
	$$
	\text{cell}(\mathbf{v}^k) = \left\{\mathbf{u} \in \mathbb{R}^m: \{k \} = \arg \min_{x \in [K]} \langle \mathbf{u}, \mathbf{v}^x \rangle - w^x \right\}.
	$$
	Here $\langle \mathbf{u}, \mathbf{v}^x \rangle$ represents the inner product of the two vectors.
	We call $\langle \mathbf{u}, \mathbf{v}^x \rangle - w^x$ the \emph{power distance} from $\mathbf{u}$ to site $\mathbf{v}^x$; thus, for every point $\mathbf{u}$ in cell$(\mathbf{v}^k)$, it holds that $\mathbf{v}^k$ is closer to $\mathbf{u}$ in power distance than any other site $\mathbf{v}^x$.
\end{definition}
Here we use a definition of power diagrams that is slightly different from the original definition in \citep{frongillo2017geometric}, which used $\Vert \mathbf{u} - \mathbf{v}^x \Vert^2 - w^x$ as the power distance. The two definitions are equivalent 
and can be transformed into one another by changing the value of $\mathbf{w}^x$. For our multi-task problem, we find it more convenient to use $\langle \mathbf{u}, \mathbf{v}^x \rangle - w^x$ as the power distance.

To give the characterization in \citet{frongillo2017geometric}, we define the following.
Let $Q_i(\r_i)$ be the set of agent $i$'s possible posterior belief of $\r_{-i}$ when his truthful report is $\r_i$.
\begin{align}
    Q_i(\r_i) = \{ \mu(\mathbf{r}_{-i}|s_i): \mu \in M, \ s_i \text{ satisfies } f_{i,\mu}(s_i) = \r_i \}. 
\end{align}
\begin{theorem}\label{thm:single_chr}
	A single-task elicitability problem $\langle \mathbf{f}, M\rangle$  is elicitable if and only if for each agent $i$, there exists a power diagram in dimension $|\mathcal{R}_{-i}|$ with $|\mathcal{R}_i|$ cells defined by $\{\mathbf{v}^{\r_i}\in \mathbb{R}^{|\mathcal{R}_{-i}|}: \r_i\in \mathcal{R}_i\}$ and $\{w^{\r_i}\in \mathbb{R}:\r_i \in \mathcal{R}_i\}$, such that each $Q_i(\r_i)$ falls into a distinct cell,
	$$
	Q_i(\r_i) \subseteq \text{cell}(\mathbf{v}^{\r_i}), \text{ for all } \r_i \in \mathcal{R}_i.  
	$$
\end{theorem}
The theorem gives the necessary and sufficient condition for a single-task problem to be elicitable. This result is given by~\citet{frongillo2017geometric} in the setting where the mechanism is minimal, i.e., the agents are asked to directly reveal their signals $\r_i = s_i$ and the support of the signals is finite. When the reports are arbitrary functions, it could be difficult to analyze the space $\Delta_{|\mathcal{R}_{-i}|}$, e.g. when the support of the reports is continuous. 

We find the following proposition to be useful in our exposition, and include a proof in the appendix. \fang{if we do not have supplementary material should we say full version instead?} Note that the result can be implied by Theorem~\ref{thm:single_chr} when report space is discrete, and we show it still holds when the support of $\r_{-i}$ is continuous.
\begin{proposition}[Robust stochastic relevancy]\label{thm:necessary_1}
A single-task problem $\langle \mathbf{f}, M\rangle$ is elicitable only if for any $i$ and $\r_i \neq \r_i' \in \mathcal{R}_i$, $Q_i(\r_i) \cap Q_i(\r_i') = \emptyset$.
\end{proposition}

If $f$ is identity, i.e. $f(s_i) = s_i$, and $M = \{\mu\}$, the above condition implies stochastic relevancy.  Thus, we call the above characterization \emph{robust stochastic relevancy}. 

Finally, if the designer's knowledge about the information structure is accurate enough, it is possible to design a strictly truthful mechanism. We defer the details to Appendix~\ref{app:pre_ball}.

\subsection{Correlated agreement mechanism} \label{sec:prelim_ca}
\citet{shnayder2016informed} proposed the \emph{correlated agreement mechanism} for multi-task peer prediction. They considered the design of \emph{minimal mechanisms}, that is, mechanisms that ask the agents to directly report their signals, so we have 
$$\r_i = s_i, \ \forall i.$$ 
The correlated agreement mechanism only needs two participants $n=2$ and two tasks $T=2$. The mechanism requires the designer to know the correlation structure of signals, but not the full signal distribution. More specifically, define the \emph{Delta matrix} $\Delta$ to be a $|\mathcal{S}_1|\times |\mathcal{S}_2|$ matrix with entry in row $s_1$ and column $s_2$ equal to
$$
\Delta[s_1, s_2] = \mu(s_1, s_2) - \mu(s_1)\mu(s_2),
$$
where $\mu(s_1, s_2)$ is the joint distribution of the two participants' signals, and $\mu(s_1), \mu(s_2)$ are the marginal distributions of $s_1$ and $s_2$ respectively. The Delta matrix describes the correlation between different realized signal values. If an entry $\Delta[s_1, s_2]>0$, then we have $\mu(s_1|s_2)> \mu(s_1)$ and $\mu(s_2|s_1)>\mu(s_2)$, which means that seeing $s_2$ will increase participant $2$'s belief about seeing $s_1$, and seeing $s_1$ will increase participant $1$'s belief about seeing $s_2$, so the signal realizations $s_1$ and $s_2$ are positively correlated.  To ensure strict truthfulness,\fang{add strict truthfulness here}  the CA mechanism requires the designer to know the sign of each entry of the Delta matrix, denoted by Sign$(\Delta)$, which means that the designer needs to know for each pair of signal realizations whether they are positively correlated or negatively correlated. The payment of the CA mechanism is designed as follows.
\begin{definition}[Correlated agreement mechanism~\citep{shnayder2016informed}]
	The correlated agreement mechanism asks the two participants to report their signals for two tasks. The payment to participant $i\in\{1,2\}$ for task $t\in\{1,2\}$ is
	$$
	p_i^{(t)}(\s^{(1:T)}) = \text{Sign}(\Delta[s_i^{(t)},s_{-i}^{(t)}]) - \text{Sign}(\Delta[s_i^{(t)}, s_{-i}^{(-t)}]),
	$$
	where $\text{Sign}(\Delta[x,y])$ represents the sign of the entry in row $x$ and column $y$ of matrix $\Delta$. 
\end{definition}
The above definition is slightly different from the original definition in~\citep{shnayder2016informed}, but they are equivalent in the sense of elicitability.
The CA mechanism will be strictly truthful if the following condition holds.
\begin{theorem}[\citet{shnayder2016informed}] \label{thm:ca_suff}
	If the matrix Sign$(\Delta)$ does not have two identical rows and it does not have two identical columns, the correlated agreement mechanism is strictly truthful.
\end{theorem}
\citet{shnayder2016informed} also discussed other incentive properties. In this work, we only consider the strict truthfulness defined in Definition~\ref{def:truthful_multi_consist}, which is equivalent to the definition of strict properness (Definition 2.6) in their work.

\subsection{Elicit predictions} \label{sec:prelim_posterior}
\citet{kong2018water} proposed a mechanism that elicits the participants' posterior about the state $\mu(\omega|s_i)$ in both single-task and multiple-task settings, when the prior $\mu(\omega)$ is known to the designer and the participants' signals are independent conditioning on the state $\omega$, that is,
$$
\mu(\omega, s_1, \dots, s_n) = \mu(\omega)\mu(s_1|\omega)\cdots \mu(s_n|\omega), \ \forall \omega, s_1, \dots, s_n.
$$ 
 \citet{chen2020truthful} further give a sufficient condition for the mechanisms to be strictly truthful. For a distribution $\mu$, define  $P^\mu_i$ to be a $|\mathcal{S}_{-i}|\times |\Omega|$ matrix with entry in row $\s_{-i}$ and column $\omega$ equal to  $\mu(\s_{-i}|\omega)$. Then we have the follows.

\begin{lemma}[\citep{kong2018water,chen2020truthful}]
Consider a single-task/multi-task problem $\langle \f, M \rangle$ with $\r_i = \mu(\omega|s_i)$. Suppose the prior $\mu(\omega)$ is known to the designer and the participants' signals are independent conditioning on the state $\omega$. If we further have rank$(P^\mu_i)=|\Omega|$ for all $\mu \in M$ and $i\in [n]$, then $\r_i = \mu(\omega|s_i)$ is elicitable by the following payment rule for every single task
$$
p_i(\r_i, \r_{-i}) = \log \left( \sum_{\omega\in \Omega} \frac{\r_i(\omega)g(\r_{-i}, \omega)}{\mu(\omega)} \right),
$$
with
$$
g(\r_{-i}, \omega) = \frac{1}{A(\r_{-i})} \cdot \frac{\prod_{j\neq i} \r_j(\omega)}{\mu(\omega)^{n-1}},
$$
where $A(\r_{-i})$ is a normalization term so that $\sum_\omega g(\r_{-i},\omega) = 1$.
\end{lemma}

\section{Multiple-task Elicitability} \label{sec:mult}

As introduced in Section~\ref{sec:single_elt}, if the designer only collects reports for one task, \citet{frongillo2017geometric} showed that a single-task peer prediction problem is elicitable if and only if each agent's posteriors can be fitted into a power diagram. But when the designer has multiple i.i.d. tasks, it is possible for the designer to exploit the similarity between the tasks and elicit information that is not elicitable in the single-task framework. For example, \citet{shnayder2016informed} showed that the Dasgupta-Ghosh mechanism \citep{dasgupta2013crowdsourced} can elicit $r_i=s_i$ when $n=2$ and the signals are categorical, which means that when an agent sees a signal, all other signals become less likely than their prior probability, i.e.,
$$
\mu(s_2 = y|s_1 = x) < \mu(s_2 = y), \forall x\neq y.
$$
 But the categorical condition clearly does not guarantee the \emph{robust stochastic relevance} (Corollary~\ref{thm:necessary_1}) without the knowledge about the marginal distribution, which is the necessary condition for a problem to be elicitable in the single-task framework. Therefore we need stronger conditions for the elicitability of multi-task problems.\fang{do we have proof or it's trivial?}

In this section, we first give a necessary and sufficient condition for a multi-task problem to be elicitable if the designer only uses \emph{scoring mechanisms}. This characterization holds when $\mathcal{R}_i$'s are finite sets, i.e., there are finitely many possible values of a participant's report. We show how to use our characterization in the setting of the CA mechanisms (Section~\ref{sec:prelim_ca}). For the general case when $\mathcal{R}_i$'s can be infinitely large, we provide necessary conditions for a multi-task problem to be elicitable. One of the necessary conditions will be the key tool that we use to obtain the results  in Section~\ref{sec:linear}. Second, for general mechanisms, we give a necessary condition for  a multi-task problem to be elicitable, assuming that the participants use consistent strategies. 

\subsection{Scoring mechanisms}

The challenge of studying the elicitability of multi-task problems is largely due to the complexity of the payment rule $\p(\r^{(1:T)})$. The payment rule $\p(\r^{(1:T)})$ can potentially be an extremely complicated function, which may not even be efficiently computable as $\r^{(1:T)}$ has exponentially many possible values. But in practice, such payments are unlikely to be appealing because of the implementation difficulty as well as the lack of transparency. So in this section, we restrict our attention to a smaller class of mechanisms: the scoring mechanisms (Definition~\ref{def:scoring}). The scoring mechanisms pay each of a participant's reports separately. The payment to a report is decided by comparing it with other participants' reports across all the tasks. To our knowledge, all of the existing strictly truthful mechanisms belong to scoring mechanisms. 

Before giving our main results, we first show that when we focus on the elicitability by scoring mechanisms, it does not matter whether the participants' strategies will be consistent or not.
\begin{theorem}\label{thm:equiv}
	When the designer only uses scoring mechanisms, a multi-task problem $\langle \f, M \rangle$ is elicitable if and only if it is elicitable under consistent strategies.
\end{theorem}
The proof of Theorem~\ref{thm:equiv} can be found in Appendix~\ref{app:mult}. In the rest of this section, we just assume that the participants can use non-consistent  strategies. The characterization of elicitability under consistent strategies is just the same.

\subsubsection{Characterization of elicitable multi-task problems}
We now give the characterization of elicitable multi-task problems when all $\mathcal{R}_i$'s are finite. 
Our first observation is that if a multi-task problem $\langle \f, M\rangle$ is elicitable by a scoring mechanism, then for each participant $i$, his posteriors should fall into correct cells of a power diagram \emph{for any given marginal distribution of other participants' truthful reports}. 
To show this formally, we introduce some notations.
Consider an agent $i$ and a given marginal distribution of other agents' truthful reports for a single task $\mu(\r_{-i})$.  Define $M_{\mu(\r_{-i})}$ to be the set of all distributions  that has marginal distribution of $\r_{-i}$ equal to $\mu(\r_{-i})$,
$$
M_{\mu(\r_{-i})} = \{ \mu' \in M: \mu'(\r_{-i}) = \mu(\r_{-i})\}
$$
 Also define $Q_i(\r_i, \mu(\r_{-i}))$ to be the set of participant $i$'s possible posteriors about $\r_{-i}$ when participant $i$'s truthful report is $\r_i$ and the marginal distribution of $\r_{-i}$ is  $\mu(\r_{-i})$,
$$
Q_i(\r_i, \mu(\r_{-i}))=\{ \mu'(\r_{-i}|s_i): \mu'\in M_{\mu(\r_{-i})},\ s_i \text{ satisfies } f_{i,\mu}(s_i) = \r_i\}.
$$ 
Then a necessary condition for a multi-task problem $\langle \f, M\rangle$ to be elicitable is that for any given $\mu(\r_{-i})$, there exists a power diagram that divides $Q_i(\r_i, \mu(\r_{-i}))$ for different $\r_i$ into different cells. We call this the power diagram constraint for given marginal distributions.
\begin{definition} \label{def:mult_nec}
	A multi-task problem $\langle \f, M\rangle$ with finite-size $\mathcal{R}_i$'s  satisfies the \emph{power diagram constraint for given marginal distributions} if for all $i$ and $\mu(\r_{-i})$, $\{ Q_i(\r_i, \mu(\r_{-i}))\}_{\r_i \in \mathcal{R}_i}$ can be fitted into a power diagram, which means that there exists a power diagram in dimension $|\mathcal{R}_{-i}|$ with $|\mathcal{R}_i|$ cells defined by $\{\mathbf{v}^{\r_i}\in \mathbb{R}^{|\mathcal{R}_{-i}|}: \r_i\in \mathcal{R}_i\}$ with associated weights $\{w^{\r_i}\in \mathbb{R}:\r_i\in \mathcal{R}_i\}$, such that each $Q_i(\r_i, \mu(\r_{-i}))$ falls into a distinct cell,
	$$
	Q_i(\r_i, \mu(\r_{-i})) \subseteq \text{cell}(\mathbf{v}^{\r_i}), \text{ for all } \r_i \in \mathcal{R}_i.  
	$$
\end{definition}
Note that checking this condition does not require the designer to know the actual marginal distribution $\mu(\r_{-i})$. The condition means that for any given marginal distribution $\mu(\r_{-i})$, if the designer restrict the possible underlying distributions to the ones that has marginal distribution of $\r_{-i}$ equal to the given $\mu(\r_{-i})$, it should be possible to fit $\{ Q_i(\r_i, \mu(\r_{-i}))\}_{\r_i \in \mathcal{R}_i}$  into a power diagram.

Below we provide an example to illustrate the sets $Q_i(\r_i, \mu(\r_{-i}))$ and how they can be fitted into power diagrams.  
\begin{example} \label{exam:simplex}
Consider the following problem instance $\langle \f, M\rangle$.  Suppose there are two agents $n = 2$, two tasks $T = 2$ and the report and signal spaces are $\mathcal{S}_1 = \mathcal{S}_2 = \mathcal{R}_1 = \mathcal{R}_2 = \{0,1,2\}$. The designer asks the participants to directly report their signals $r_i = f_{i,\mu}(s_i) = s_i$. The set of possible distributions $M$ is the set of all distributions that have the sign of the Delta matrix Sign($\Delta$) (defined in Section~\ref{sec:prelim_ca}) equal to $$\begin{bmatrix}1&-1&-1\\-1&1&-1\\-1&-1&1\end{bmatrix}.$$ 
This means that for a participant $i$, seeing a signal realization $s_i^*\in \{0,1,2\}$ will increase the probability that the other participant also observes the same signal realization $s_i^*$ but decrease the probability that the other participant observes a different signal realization.  
Given such $M$ and $\mathcal{R}_i = \{0,1,2\}$, we use simplex plot on $\Delta(\{0,1,2\})$ to illustrate the sets $Q_i(r_i, \mu(r_{-i}))$. In Figure~\ref{fig:simplex}, each point on the simplex plot represents a distribution on $\{0,1,2\}$. We choose two marginal distributions, $\mu(r_{-i}) = (1/3,1/3,1/3)$ on the left and $\mu(r_{-i}) = (0.2,0.3,0.5)$ on the right. The colored areas are the set $\{Q_i(r_i, \mu(r_{-i}))\}_{r_i\in \mathcal{R}_i}$, and the dashed lines are the boundary of the cells of a power diagram 
	with certain sites and weights. The exact definition of the sites and the weights will be given in \eqref{eq:ca_sites} later in Section~\ref{sec:main_suff} when we discuss the  application of our results.  
	
\end{example}
\begin{figure}
\centering
\includegraphics[width = .32\textwidth]{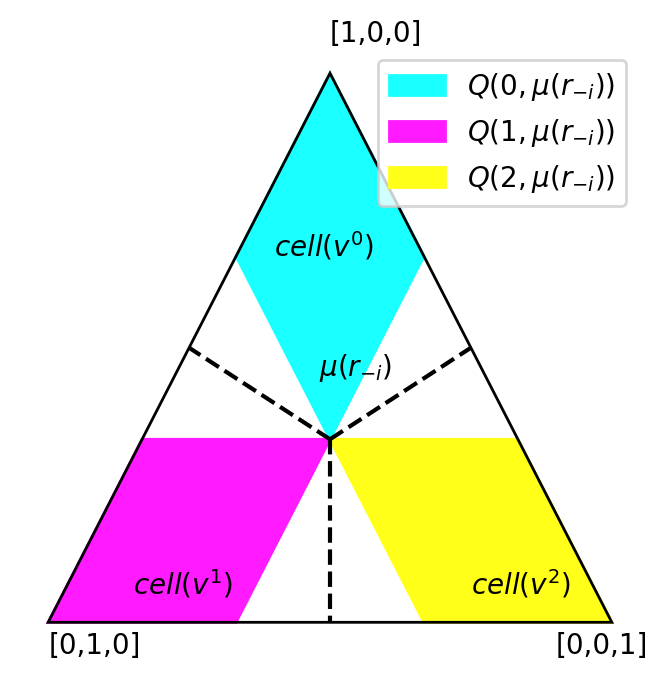}
\qquad
\includegraphics[width = .32\textwidth]{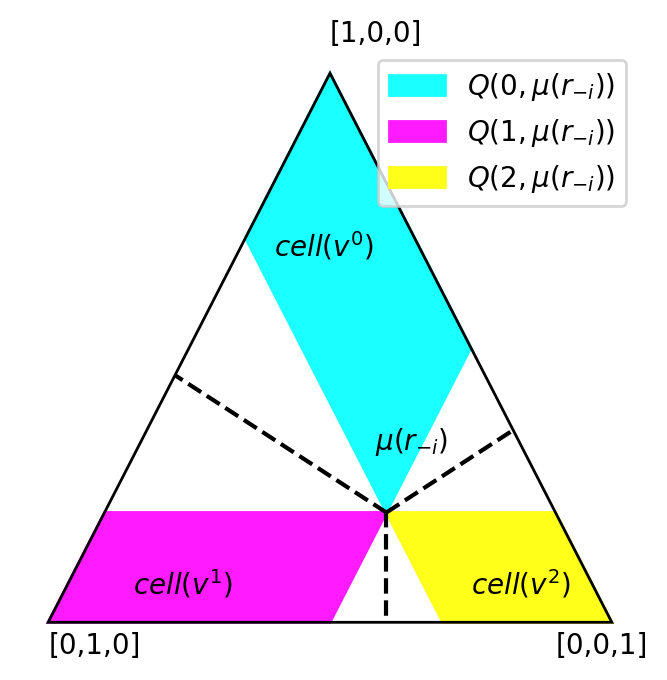}
\caption{With various marginal distributions $\mu(r_{-i})$ and $r_i = 0,1,2$, the simplex plots illustrate the sets $\{Q_i(r_i, \mu(r_{-i}))\}_{r_i\in \mathcal{R}_i}$ (colored areas) defined in Example~\ref{exam:simplex} and the cells of power diagrams cell$(\mathbf{v}^{r_i})$ (with dashed boundary)  that separate them.  In our example, the point at the intersection of the dashed lines is the marginal distributions $\mu(r_{-i})$. \SR{May need to change the notation}}
\label{fig:simplex}
\end{figure}

Now assume that $\langle \f, M\rangle$ satisfies the power diagram constraint for given \fang{add any}\SR{this is the term that's defined in the previous definition}\fang{I see} marginal distributions. For any report $\r_i \in \mathcal{R}_i$, denote the site $\mathbf{v}^{\r_i}$ for $\mu(\r_{-i})$ by $\mathbf{v}^{\r_i}(\mu(\r_{-i}))$ and denote the associated weight by $w^{\r_i}(\mu(\r_{-i}))$. Here we abuse the notation and consider $w^{r_i}:\Delta(\mathcal{R}_{-i}) \to \mathbb{R}$ and $\mathbf{v}^{\r_i}:\Delta(\mathcal{R}_{-i}) \to \mathbb{R}^{|\mathcal{R}_{-i}|}$ as functions of the marginal distribution $\mu(\r_{-i})$.
Our second observation is that the sites of the power diagrams $\mathbf{v}^{\r_i}(\mu(\r_{-i}))$ and the associated weights $w^{\r_i}(\mu(\r_{-i}))$ need to be affine functions of the marginal distribution of other participants' reports for $T-1$ tasks $\mu(\r_{-i}^{(1:T-1)})$.\footnote{Because the tasks are i.i.d., the distribution for $T-1$ tasks $\mu(\r_{-i}^{(1:T-1)})$ can be generated by the distribution for a single task $\mu(\r_{-i})$.} To be more specific, let $\mathbf{u}$ be the length-$|\mathcal{R}_{-i}|$ vector that represents the distribution $\mu(\r_{-i})$. Then $\mu(\r_{-i}^{(1:T-1)})$ can be represented by the $(T-1)$-th tensor power of the vector, $\mathbf{u}^{\otimes (T-1)}$. Our observation is that both $\mathbf{v}^{\r_i}(\mathbf{u})$ and $w^{\r_i}(\mathbf{u})$ need to be affine functions of $\mathbf{u}^{\otimes (T-1)}$.
Moreover, when such power diagrams and affine functions exist, we can find a mechanism that elicits $\r$. 

\begin{theorem} \label{thm:mult_suff}
	A multi-task problem $\langle \f, M\rangle$ with finite-size $\mathcal{R}_i$'s is elicitable by scoring mechanisms if and only if 
	\begin{enumerate}
	\item 	It satisfies the power diagram constraint for given marginal distributions. Let $\mathbf{u}$ be a vector that represents a marginal distribution $\mu(\r_{-i})$ and denote by $\mathbf{v}^{\r_i}(\mathbf{u})$ and $w^{\r_i}(\mathbf{u})$ the sites and the weights of the power diagram for the marginal distribution $\mathbf{u}$.
	\item Furthermore, 
	for every $i\in[n]$ and every $\r_i \in \mathcal{R}_i$, there exist a  matrix $\mathbf{D}_{\r_i}$ with $|\mathcal{R}_{-i}|$ rows and $|\mathcal{R}_{-i}|^{T-1}$ columns and a vector $\mathbf{e}_{\r_i}\in \mathbb{R}^{|\mathcal{R}_{-i}|}$ with
	$$
	\mathbf{v}^{\r_i}(\mathbf{u}) = \mathbf{D}_{\r_i} \cdot \mathbf{u}^{\otimes (T-1)} + \mathbf{e}_{\r_i}, \quad \text{for all possible } \mathbf{u},
	$$
	and there exists a vector $\mathbf{h}_{\r_i}$ with length $|\mathcal{R}_{-i}|^{T-1}$ such that 
	$$
	w^{\r_i}(\mathbf{u}) = \mathbf{h}_{\r_i}^\top \cdot \mathbf{u}^{\otimes (T-1)}, \quad \text{for all possible } \mathbf{u}.
	$$
	\end{enumerate}
	Moreover, if such power diagrams and affine functions exist, we can find a mechanism that is strictly truthful for any $\mu\in M$ with payments defined by the entries of  $\, \mathbf{D}, \mathbf{e}$ and $\mathbf{h}$ as
	\begin{align} \label{eqn:def_pi}
p_i^{(t)}(\r_i,  \r_{-i}^{(t)}, \r_{-i}^{(-t)}) = -\mathbf{D}_{\r_i}[\r_{-i}^{(t)}, \r_{-i}^{(-t)}] \ - \mathbf{e}_{\r_i}[\r_{-i}^{(t)}] +\mathbf{h}_{\r_i}[\r_{-i}^{(-t)}], \quad \forall \r_i \in \mathcal{R}_i,\ \r_{-i}^{(t)},\ \r_{-i}^{(-t)}.
	\end{align}
\end{theorem}
We want to point out that the parameters $\mathbf{e}, \mathbf{h}$ in the theorem can actually be merged into matrix $\mathbf{D}$ so that $\mathbf{D}_{\r_i} = \mathbf{e}_{\r_i}\cdot \mathbf{1}^\top+\mathbf{1}\cdot \mathbf{h}_{\r_i}^\top$. We use this form because in some applications, it is more convenient to separate $\mathbf{e}$ and $\mathbf{h}$ and have a payment in form~\eqref{eqn:def_pi}, for example when we apply the theorem to the setting of the correlated agreement mechanism~\citep{shnayder2016informed}.

\begin{proof}[Proof of Theorem~\ref{thm:mult_suff}]
We consider the general setting when the agents can use non-consistent strategy over different tasks. 

We first prove the necessity of the condition. 
Suppose that the designer has a strictly truthful scoring mechanisms for a multi-task problem $\langle \f, M\rangle$, so that the payment has the form
\begin{align}
		p_i(\tilde{\r}_i^{(1:T)},\ \tilde{\r}_{-i}^{(1:T)}) = \sum_{t=1}^T p_i^{(t)}(\tilde{\r}_i^{(t)}, \tilde{\r}_{-i}).
\end{align}	
 When participant $i$ decides the strategy for task $t$, his report for this task $\tilde{\r}_i^{(t)}$ only affects the payment for this task $p_i^{(t)}(\tilde{\r}_i^{(t)}, \tilde{\r}_{-i})$. Therefore his best strategy for task $t$ 
 is the one that maximizes the expected payment for task $t$ assuming that the other agents truthfully report, that is,
$
\E_\mu[p_i^{(t)}(\tilde{\R}_i^{(t)}, \R_{-i})].
$
This means that for each signal realization $s_i^{(t)}$,  participant $i$ should choose the report $\tilde{\r}_i^{(t)}$ that will maximize the conditional expectation $\E_\mu[p_i^{(t)}(\tilde{\R}_i^{(t)}, \R_{-i})|s_i^{(t)}]$. Since the tasks are independent, observing $s_i^{(t)}$ will only change agent $i$'s belief about other people's reports for this task $\r_{-i}^{(t)}$, but not the reports for other tasks $\r_{-i}^{(-t)}$. So we can factor out $\r_{-i}^{(t)}$ in the conditional expectation as follows
\begin{align}
\E_\mu\big[p_i^{(t)}\big(\tilde{\R}_i^{(t)},\, \R_{-i}^{(t)},\, \R_{-i}^{(-t)}\big)|s_i^{(t)}\big] & = \sum_{\r_{-i}^{(t)}, \r_{-i}^{(-t)}} \mu(\r_{-i}^{(t)}, \r_{-i}^{(-t)}|s_i^{(t)})\cdot p_i^{(t)}\big(\tilde{\r}_i^{(t)},\, \r_{-i}^{(t)},\, \r_{-i}^{(-t)}\big) \notag \\
& = \sum_{\r_{-i}^{(t)}, \r_{-i}^{(-t)}} \mu(\r_{-i}^{(t)}|s_i^{(t)})\cdot \mu(\r_{-i}^{(-t)})\cdot p_i^{(t)}\big(\tilde{\r}_i^{(t)},\, \r_{-i}^{(t)},\, \r_{-i}^{(-t)}\big)\notag \\
& = \sum_{\r_{-i}^{(t)}} \mu(\r_{-i}^{(t)}|s_i^{(t)})\cdot \sum_{\r_{-i}^{(-t)}} \mu(\r_{-i}^{(-t)})\cdot p_i^{(t)}\big(\tilde{\r}_i^{(t)},\, \r_{-i}^{(t)},\, \r_{-i}^{(-t)}\big). \label{eqn:main_pf_eq1}
\end{align}
Now consider a given marginal distribution of other participants' report for one task, $\mu(\r_{-i})$. Define $\alpha_{y}( \r_{-i}^{(t)})$ to be the value of the second sum in~\eqref{eqn:main_pf_eq1} when participant $i$ reports $y$, that is, $\tilde{\r}_i^{(t)} = y$, and other participants' truthful report for task $t$ is $\r_{-i}^{(t)}$,
$$
\alpha_{y}( \r_{-i}^{(t)}) = \sum_{\r_{-i}^{(-t)}} \mu(\r_{-i}^{(-t)})\cdot p_i^{(t)}\big(y,\, \r_{-i}^{(t)},\, \r_{-i}^{(-t)}\big),
$$
where the distribution $\mu(\r_{-i}^{(-t)})$ is generated according to the given marginal distribution $\mu(\r_{-i})$.
Then 
participant $i$'s expected payoff (for task $t$) conditioning on $s_i^{(t)}$ when reporting $y$ is the inner product of the vectors that represent posterior belief $\mu(\r_{-i}^{(t)}|s_i^{(t)})$ and $\alpha_{y}( \r_{-i}^{(t)})$, 
$$
\E_\mu\big[p_i^{(t)}\big(y,\, \R_{-i}^{(t)},\, \R_{-i}^{(-t)}\big)|s_i^{(t)}\big] =\sum_{\r_{-i}^{(t)}} \mu(\r_{-i}^{(t)}|s_i^{(t)})\cdot \alpha_{y}( \r_{-i}^{(t)}) = \langle \bm{\mu}(\r_{-i}^{(t)}|s_i^{(t)}), \bm{\alpha}_{y}\rangle.
$$
Here we abuse the notation to use $\bm{\mu}(\r_{-i}^{(t)}|s_i^{(t)})$ to represent a vector.
As a result, by the definition of strict truthfulness, when agent $i$'s truthful report is $\r_i$, his expected payoff (for task $t$) should be uniquely maximized when reporting $\r_i$, which means 
\begin{align}\label{eqn:mult_power}
\{\r_i\} = \arg\max_{y\in \mathcal{R}_i} \langle \bm{\mu}(\r_{-i}|s_i), \bm{\alpha}_y\rangle.
\end{align}
We construct a power diagram by setting $\mathbf{v}^y = -\bm{\alpha}_y$ and $w^y = 0$. 
Then \eqref{eqn:mult_power} is equivalent to 
$$
\mu(\r_{-i}|s_i) \in \text{cell}(\mathbf{v}^{\r_i}).
$$
Therefore if $\langle \f, M \rangle$ is elicitable by scoring mechanisms, $\{ Q_i(\r_i, {\mu(\r_{-i})})\}_{r_i \in \mathcal{R}_i}$ can be fitted into a power diagram.
Furthermore, let matrix $\mathbf{D}_y$ be a $|\mathcal{R}_{-i}|\times |\mathcal{R}_{-i}|^{T-1}$ matrix with $p_i^{(t)}(y, \r_{-i}^{(t)}, \r_{-i}^{(-t)})$ in row $\r_{-i}^{(t)}$ and column $\r_{-i}^{(-t)}$. Then by our construction of $\mathbf{v}^y$, we have 
\begin{align*}
	\mathbf{v}^y(\mu(\r_{-i})) = \alpha_{y} = \mathbf{D}_y \cdot \mu(\r_{-i}^{(-t)}).
\end{align*}
Therefore it is necessary that there exists $\mathbf{v}^y(\mu(\r_{-i}))$ which is an affine function of $\mu(\r_{-i}^{(-t)})$. 

Next, we prove the sufficiency of the condition.
If we have power diagrams defined by $\{\mathbf{v}^{\r_i}(\mu(\r_{-i}))\in \mathbb{R}^{|\mathcal{R}_{-i}|}: \r_i\in \mathcal{R}_i\}$ and $w(\r_i)$ such that for any given $\mu(\r_{-i})$, $Q_i(\r_i, \mu(\r_{-i}))$ falls into a distinct cell 
	\begin{align} \label{eqn:mult_cond}
	Q_i(\r_i, \mu(\r_{-i})) \subseteq \text{cell}(\mathbf{v}^{\r_i}(\mu(\r_{-i}))), \text{ for all } \r_i \in \mathcal{R}_i.  
	\end{align}
	and for each $\r_i$, there exist a  matrix $\mathbf{D}_{\r_i}$ with $|\mathcal{R}_{-i}|$ rows and $|\mathcal{R}_{-i}|^{T-1}$ columns and a vector $\mathbf{e}_{\r_i}\in \mathbb{R}^{|\mathcal{R}_{-i}|}$ such that 
	$$
	\mathbf{v}^{\r_i}(\mu(\r_{-i})) = \mathbf{D}_{\r_i} \cdot \mu(\r_{-i}^{(-t)}) + \mathbf{e}_{\r_i},
	$$
	and there exists a vector $\mathbf{h}$ with length $|\mathcal{R}_{-i}|^{T-1}$ such that 
	$$
	w^{\r_i}(\mu(\r_{-i})) = \mathbf{h}_{\r_i}^\top \cdot \mu(\r_{-i}^{(-t)}).
	$$
Consider the payment rule~\eqref{eqn:def_pi} given by the theorem
\begin{align*} 
p_i^{(t)}(\r_i,  \r_{-i}^{(t)}, \r_{-i}^{(-t)}) = -\mathbf{D}_{\r_i}[\r_{-i}^{(t)}, \r_{-i}^{(-t)}] \ - \mathbf{e}_{\r_i}[\r_{-i}^{(t)}] +\mathbf{h}_{\r_i}[\r_{-i}^{(-t)}].
\end{align*}
Then for any underlying distribution $\mu$, suppose participant $i$ observes $s_i^{(t)}$ for task $t$, if participant $i$ reports $y \in \mathcal{R}_i$, as shown in~\eqref{eqn:main_pf_eq1}, his expected payment for task $t$ is
\begin{align*}
\E_\mu\big[p_i^{(t)}\big(y,\, \R_{-i}^{(t)},\, \R_{-i}^{(-t)}\big)|s_i^{(t)}\big] 
& = \sum_{\r_{-i}^{(t)}} \mu(\r_{-i}^{(t)}|s_i^{(t)})\cdot \sum_{\r_{-i}^{(-t)}} \mu(\r_{-i}^{(-t)})\cdot p_i^{(t)}\big(y,\, \r_{-i}^{(t)},\, \r_{-i}^{(-t)}\big)\\
& = -\langle \bm{\mu}(\r_{-i}^{(t)}|s_i^{(t)}),\, \mathbf{D}_{y} \cdot \bm{\mu}(\r_{-i}^{(-t)}) + \mathbf{e}_{y} \rangle + w^{y}(\mu(\r_{-i}))\\
& = -\langle \bm{\mu}(\r_{-i}^{(t)}|s_i^{(t)}), \,	\mathbf{v}^{y}(\mu(\r_{-i}))\rangle + w^{y}(\mu(\r_{-i})).
\end{align*}
Consequently, the payment rule will be strictly truthful because by the condition~\eqref{eqn:mult_cond} we have 
$$
\mu(\r_{-i}|s_i) \in Q_i(\r_i, {\mu(\r_{-i})}) \subseteq \text{cell}(\mathbf{v}^{\r_i}(\mu(\r_{-i}))),
$$
which by the definition of $\text{cell}(\mathbf{v}^{\r_i}(\mu(\r_{-i})))$ (Definition~\ref{def:power_diag}) means that
\begin{align*}
\{\r_i\} & = \arg\min_{y\in \mathcal{R}_i} \langle \mathbf{v}^y(\mu(\r_{-i})),\, \bm{\mu}(\r_{-i}|s_i)\rangle - w^{y}(\mu(\r_{-i}))\\
& = \arg \min_{y\in \mathcal{R}_i} -E_\mu\big[p_i^{(t)}\big(y,\, \R_{-i}^{(t)},\, \R_{-i}^{(-t)}\big)|s_i^{(t)}\big] 
 \\
& = \arg \max_{y\in \mathcal{R}_i}\ \E_\mu\big[p_i^{(t)}\big(y,\, \R_{-i}^{(t)},\, \R_{-i}^{(-t)}\big)|s_i^{(t)}\big].
\end{align*}
\end{proof}

Theorem~\ref{thm:mult_suff} applies to the general scoring mechanisms whose  $p_i^{(t)}(\tilde{\r}_i^{(t)}, \tilde{\r}_{-i}^{(1:T)})$ can be an arbitrary function of $\tilde{\r}_i^{(t)}$ and $\tilde{\r}_{-i}^{(1:T)}$. Since $\tilde{\r}_{-i}^{(1:T)}$ has exponentially many possible values, the payment $p_i^{(t)}(\tilde{\r}_i^{(t)}, \tilde{\r}_{-i}^{(1:T)})$ of an arbitrary scoring mechanism can possibly be hard to compute. In practice, for computational reasons, the designer may only want to decide payments based on a function of $\tilde{\r}_{-i}^{(1:T)}$, for example, a sufficient statistic of $\tilde{\r}_{-i}^{(1:T)}$. In this case, we can have a simplified version of Theorem~\ref{thm:mult_suff} with smaller $\mathbf{D}_{\r_i}$ and $\mathbf{h}_{\r_i}$, which will have polynomial sizes if the function of $\tilde{\r}_{-i}^{(1:T)}$ has polynomially many possible values. We defer the details to Appendix~\ref{app:main_simplified}.

Our characterization can be used in two ways. 
\subsubsection{The sufficiency of the condition} \label{sec:main_suff}
 Given a problem instance $\langle \f, M\rangle$, if we are able to find power diagrams as described in Theorem~\ref{thm:mult_suff}, then the payment rule~\eqref{eqn:def_pi} will guarantee strict truthfulness. 
To give an example, we apply Theorem~\ref{thm:mult_suff} to the setting of the correlated agreement mechanism (Theorem~\ref{thm:ca_suff}). We show how to find power diagrams as described in Theorem~\ref{thm:mult_suff} when the designer knows the sign of the Delta matrix Sign$(\Delta)$ and the matrix Sign$(\Delta)$ does not have two identical rows or two identical columns. We will find that the payment rule~\eqref{eqn:def_pi} induced by the power diagrams is just the payment rule of the CA mechanism.

Consider the setting for the CA mechanism in which $n=2, T=2$ and $r_i = s_i$ for all $i$.
If the designer knows Sign$(\Delta)$ and Sign$(\Delta)$ does not contain two identical rows or columns, define Sign$(\Delta[s_1,:])\in\{-1,+1\}^{|\mathcal{S}_2|}$ to be the sign of the $s_1$-th row of $\Delta$. Our goal is to find power diagrams that will separate a participant $i$'s possible posteriors of the other participant's report $\mu(s_2|s_1)$ for different signal realizations $s_1$, for any given marginal distribution $\mu(s_2)$. It may not be straightforward if we directly try to separate possible $\mu(s_2|s_1)$ for different $s_1$, but we can easily find power diagrams that separate $\mu(s_2|s_1)-\mu(s_2)$ for different $s_1$. By the definition of the Delta matrix $\Delta$, $\mu(s_2|s_1)-\mu(s_2)$ is just the $s_1$-th row of $\Delta$. If the designer knows the sign of $\Delta$, we can just define the sites of a power diagram to be the rows of $-$Sign$(\Delta)$, that is, define 
$$
\tilde{\mathbf{v}}^{s_1} = -\text{Sign}(\Delta[s_1, :]),\quad \tilde{w}^{s_1} = 0, \quad \forall s_1\in \mathcal{S}_1.
$$
Then for a signal realization $s_1 \in \mathcal{S}_1$ and for any possible $\mu(s_2|s_1)-\mu(s_2)$,  the site $\tilde{\mathbf{v}}^{s_1}$ will be the closest to $\mu(s_2|s_1)-\mu(s_2)$ among all the sites in terms of power distance, that is,
\begin{align}\label{eqn:ca_power}
	\{s_1\} &= \arg \min_{y\in \mathcal{S}_1}\ \langle -\text{Sign}(\Delta[y,:]), \bm{\mu}(s_2|s_1)-\bm{\mu}(s_2) \rangle \notag\\
	& = \arg \min_{y\in \mathcal{S}_1}\ \langle -\text{Sign}(\Delta[y,:]), \bm{\mu}(s_2|s_1) \rangle + \langle \text{Sign}(\Delta[y,:]), \bm{\mu}(s_2)\rangle.
\end{align}
Here we abuse the notation of $\bm{\mu}(s_2|s_1)$ and $\bm{\mu}(s_2)$ to denote the vectors that represent the distribution.
This can be converted into power diagrams that will separate possible $\mu(s_2|s_1)$ for different $s_1$. From~\eqref{eqn:ca_power}, we can define
\begin{equation}\label{eq:ca_sites}
\mathbf{v}^{s_1}(\mu(\s_2)) = -\text{Sign}(\Delta[s_1, :])\text{, and  }w^{s_1}(\mu(\s_2)) =- \langle \text{Sign}(\Delta[s_1, :]), \bm{\mu}(s_2)\rangle.
\end{equation}
Fortunately, both $\mathbf{v}^{s_1}(\mu(\s_2))$ and $w^{s_1}(\mu(\s_2))$ are affine in $\mu(s_2)$, with parameters
\begin{align*}
\mathbf{D}_{s_1} = \mathbf{0},\quad \mathbf{e}_{s_1} = -\text{Sign}(\Delta[s_1, :]), \quad \mathbf{h}_{s_1} = -\text{Sign}(\Delta[s_1, :]).
\end{align*}
Then the payment defined by Theorem~\ref{thm:mult_suff} $$
p_i^{(t)}(r_i,  \r_{-i}^{(t)}, \r_{-i}^{(-t)}) = -\mathbf{D}_{r_i}[\r_{-i}^{(t)}, \r_{-i}^{(-t)}] \ - \mathbf{e}_{r_i}[\r_{-i}^{(t)}] +\mathbf{h}_{r_i}[\r_{-i}^{(-t)}]$$
gives 
$$
p_1^{(t)}(\s) = \text{Sign}(\Delta[s_1^{(t)}, s_2^{(t)}]) - \text{Sign}(\Delta[s_1^{(t)}s_2^{(-t)}]),
$$
which is just the payment rule of the CA mechanism.



\subsubsection{The necessity of the condition} If a problem instance $\langle \f, M\rangle$ does not satisfy the condition in our theorem, then one should not hope for a strictly truthful scoring mechanism. The designer should seek  additional knowledge about the distribution or elicit different information. Testing whether there exist power diagrams as described in Theorem~\ref{thm:mult_suff} may not be easy. Below we provide two simpler necessary conditions for the existence of strictly truthful scoring mechanisms. When the report space $\mathcal{R}_i$ is discrete, these two necessary conditions can be implied by Theorem~\ref{thm:mult_suff}. In addition, they will also hold for continuous report space, for example, when $\mathcal{R}_i = \Delta \Omega$ is the space of all posteriors of the participant. 

First, observe that a cell of a power diagram must be convex. Therefore, if we fix a marginal distribution $\mu(\r_{-i})$, $Q_i(\r_i, \mu(\r_{-i}))$ for different $\r_i$ should fall into disjoint convex sets. More specifically, we should have the following. 
\begin{proposition}\label{coro:convex}
A multi-task problem $\langle \f, M\rangle$ is not elicitable if there exist $i\in[n]$ and a marginal distribution $\mu(\r_{-i})$ and $\r_i \neq \r_i' \in \mathcal{R}_i$, such that 
there exist $k,l \in \mathbb{Z}^+$ and $\mathbf{q}_1, \dots, \mathbf{q}_k \in  Q_i(\r_i, \mu(\r_{-i}))$ and $\beta_1, \dots, \beta_k\in [0,1]$ with $\beta_1 + \cdots +\beta_k = 1$ and $\mathbf{q}'_1, \dots, \mathbf{q}'_l \in  Q_i(\r_i', \mu(\r_{-i}))$ and $\beta_1', \dots, \beta_l'\in [0,1]$ with $\beta_1' + \cdots +\beta_l' = 1$ such that
$$
\beta_1 \cdot \mathbf{q}_1 + \cdots + \beta_k \cdot \mathbf{q}_k = \beta_1' \cdot \mathbf{q}_1' + \cdots + \beta_l' \cdot \mathbf{q}_l' .
$$
\end{proposition}
The second necessary condition is even simpler. It is similar to the robust stochastic relevance condition~\ref{thm:necessary_1} that we proposed for the single-task peer prediction. The difference in the multi-task setting is that we need to fix a marginal distribution $\mu(\r_{-i})$.
\begin{proposition}[Robust stochastic relevance with given marginal distributions] \label{coro:marginal}
	For a multi-task problem $\langle \f, M\rangle$, if there exist $i\in[n]$ and a marginal distribution $\mu(\r_{-i})$ such that
	$$
	Q_i(\r_i, \mu(\r_{-i})) \cap Q_i(\r_i', \mu(\r_{-i})) \neq \emptyset
	$$
	for $\r_i\neq \r_i' \in \mathcal{R}_i$, then $\langle \f, M\rangle$ is not elicitable.
\end{proposition}
Although Proposition~\ref{coro:marginal} seems apparent based on Theorem~\ref{thm:mult_suff}, it can be useful when we consider a complicated report space, e.g., when we consider eliciting linear properties of the participants' posteriors (Section~\ref{sec:linear}). It is the main tool that we use to prove our main results in Section~\ref{sec:linear}.

\subsection{General mechanisms}
Although all the existing strictly truthful mechanisms that we know belong to scoring mechanisms, recently there have been attempts to design more complicated mechanisms with good properties, e.g., the mechanism proposed by \citet{kong2020dominantly} is not a scoring mechanism. In this section, we look at  general multi-task mechanisms and give a necessary condition for a multi-task problem to be elicitable, assuming that the participants use consistent strategies, which is a commonly used assumption by the multi-task peer prediction literature. 

The necessary condition basically says that, given a joint distribution of the participants' reports without naming a participant $i$'s report, the designer should at least be able to label participant $i$'s report based on the distribution. For example, consider two participants who are asked to report a high signal or a low signal, i.e., $\mathcal{R}_1 = \mathcal{R}_2 = \{\text{high}, \text{low} \}$. Then given a possible joint distribution without the labels of the first participant's report, e.g. the distribution represented by the table in the left. The designer should be able to tell which of $r_1$ and $r_1'$ is the high report.
\begin{table}[!h]
\centering
\begin{tabular}{|c|c|c|}
	\hline 
	& high & low\\
	\hline 
	$r_1$ & $0.8$ & $0.2$ \\
	\hline 
	$r_1'$ & $0.2$ & $0.8$\\
	\hline
\end{tabular}
$\quad \Longrightarrow \quad $
\begin{tabular}{|c|c|c|}
	\hline 
	& high & low\\
	\hline 
	high & $0.8$ & $0.2$ \\
	\hline 
	low & $0.2$ & $0.8$\\
	\hline
\end{tabular}
\end{table}

Formally, we have the following theorem.
\begin{theorem}
A multi-task problem $\langle \f, M\rangle$ is elicitable under consistent strategies only if for any distribution $\mu\in M$, any agent $i$ and any permutation $\pi$ of $\mathcal{R}_i$, there exists no $\tilde{\mu}\in M$ such that 
$$
\mu(\pi(r_i), \r_{-i}) = \tilde{\mu}(r_i, \r_{-i}).
$$ 	
Or equivalently, for any $i$, if we use a $|\mathcal{R}_i| \times |\mathcal{R}_{-i}|$ matrix $A_\mu$ with $\mu(r_i, \r_{-i})$ in row $r_i$ and column $\r_{-i}$ to represent any $\mu\in M$, then for any $\mu\in M$ and  any $|\mathcal{R}_i|\times |\mathcal{R}_i|$ permutation matrix $P$, there exists no $\tilde{\mu} \in M$ such that 
$$
P\cdot A_\mu = A_{\tilde{\mu}}.
$$  
\end{theorem}
\begin{proof}
	We prove that if there exist $\mu, \tilde{\mu}\in M$ and permutation $\pi$ that have 
	$
	\mu(\pi(r_i), \r_{-i}) = \tilde{\mu}(r_i, \r_{-i}),
	$
	there exists no payment rule $p_i(r_i^{(1:T)}, \r_{-i}^{(1:T)})$ that guarantees strict truthfulness under consistent strategies for both $\mu$ and $\tilde{\mu}$. 
	
	Since $\mu(\pi(r_i), \r_{-i}) = \tilde{\mu}(r_i, \r_{-i}),$ we have 
	\begin{align*}
	\E_{\tilde{\mu}}[p_i(\pi(r_i^{(1:T)}), \r_{-i}^{(1:T)})] & = \sum_{\r^{(1:T)}} \tilde{\mu}(r_i^{(1:T)}, \r_{-i}^{(1:T)}) \cdot p_i(\pi(r_i^{(1:T)}), \r_{-i}^{(1:T)}) \\
	& = \sum_{\r^{(1:T)}} \mu(\pi(r_i^{(1:T)}), \r_{-i}^{(1:T)}) \cdot p_i(\pi(r_i^{(1:T)}), \r_{-i}^{(1:T)}) \\
	& = \sum_{\r^{(1:T)}} \mu(r_i^{(1:T)}, \r_{-i}^{(1:T)}) \cdot p_i(r_i^{(1:T)}, \r_{-i}^{(1:T)}) \\
	&= \E_{\mu}[p_i(r_i^{(1:T)}, \r_{-i}^{(1:T)})],
	\end{align*}
	which means that the expected payment of permuting the reporting when the underlying distribution is $\tilde{\mu}$ is equal to the expected payment of truthfully reporting when the distribution is $\mu$. At the same time, it should also hold that
	\begin{align*}
		\E_{\tilde{\mu}}[p_i(r_i^{(1:T)}, \r_{-i}^{(1:T)})]  & = \sum_{\r^{(1:T)}} \tilde{\mu}(r_i^{(1:T)}, \r_{-i}^{(1:T)}) \cdot p_i(r_i^{(1:T)}, \r_{-i}^{(1:T)}) \\
		& = \sum_{\r^{(1:T)}} \mu(\pi(r_i^{(1:T)}), \r_{-i}^{(1:T)}) \cdot p_i(r_i^{(1:T)}, \r_{-i}^{(1:T)}) \\
		& = \sum_{\r^{(1:T)}} \mu(r_i^{(1:T)}, \r_{-i}^{(1:T)}) \cdot p_i(\pi^{-1}(r_i^{(1:T)}), \r_{-i}^{(1:T)}) \\
		& = \E_{\mu}[p_i(\pi^{-1}(r_i^{(1:T)}), \r_{-i}^{(1:T)})],
	\end{align*}
 	which means that the expected payment of truthfully reporting when the underlying distribution is $\tilde{\mu}$ is equal to the expected payment of inversely permuting the strategy when the distribution is $\mu$. Then it is easy to see that $p_i(\cdot)$ cannot be  strictly truthful for both $\mu$ and $\tilde{\mu}$.
	If $p_i(r_i^{(1:T)}, \r_{-i}^{(1:T)})$ is strictly truthful for $\tilde{\mu}$, then by definition, permuting the truthful reports should lead to strictly lower expected payment,
	$$
	\E_{\tilde{\mu}}[p_i(r_i^{(1:T)}, \r_{-i}^{(1:T)})]>\E_{\tilde{\mu}}[p_i(\pi(r_i^{(1:T)}), \r_{-i}^{(1:T)})],
	$$
	which will violate the strict truthfulness for $\mu$
	$$
	\E_{\mu}[p_i(\pi^{-1}(r_i^{(1:T)}), \r_{-i}^{(1:T)})] > \E_{\mu}[p_i(r_i^{(1:T)}, \r_{-i}^{(1:T)})]
	$$
	according to the two equalities above. 
\end{proof}
One may wonder whether the permutation can be replaced by any reporting strategy, i.e., whether the permutation matrix can be replaced by any Markov matrix. The answer is negative. For the proof to hold, we need the matrix $P$ to be invertible, and the only invertible Markov matrices are permutation matrices. 

The previous works that design strictly truthful multi-task peer prediction mechanisms make assumptions that automatically satisfy this condition. 
To our knowledge, this necessary condition did not appear in any previous work, but similar techniques have been used in a different setting. \citet{shnayder2016informed} used a similar approach to prove that the CA mechanism is maximally strong truthful among a broader class of mechanisms (Theorem 5.9).

\section{Linear Properties} \label{sec:linear}
In this section, we consider the reports that are linear in the participant's posterior $\mu(\omega|s_i)$. 
More specifically, the report function $\r_i$ is a length-$L$ vector with  $$\r_i = \G \cdot \bm{\mu}(\omega|s_i)$$ where $\G$ is a $L\times|\Omega|$ matrix that represents the linear transformation from the posterior $\mu(\omega|s_i)$ to $\r_i$. Or equivalently, each entry of $\r_i$ is the expectation of a random variable defined on $\Omega$. Common examples include
\begin{itemize}
\item the posterior itself $\ \r_i = \mu(\omega|s_i)$;
\item the moments of the state $\omega$ when the state is a real number $\E[\omega|s_i], \E[\omega^2|s_i], \E[\omega^3|s_i] \dots$.
\end{itemize}
The elicitability of such linear properties has been studied by \citet{abernethy2012characterization} when the designer can observe $\omega$ and design a payment based on both the report and $\omega$. They showed that linear properties are always elicitable in that case. But in peer prediction problems, the designer cannot observe $\omega$ but only have the participants' reports that are correlated with $\omega$. 

Our results show that it may not always be possible to elicit reports that are linear in the participants' posterior. We first consider two participants with signals independent conditioning on $\omega$. In this case, if the designer has a single-task problem or the designer only considers scoring mechanisms for a multi-task problem, then the designer basically can only elicit linear properties that are equivalent to $\mu(\omega|s_i)$, assuming that the designer is uncertain about the underlying distribution. We then look into the case when the reports $r_i$ is just the participants' posteriors $\mu(\omega|s_i)$. We give a necessary condition for $\mu(\omega|s_i)$ to be elicitable. This condition implies that the mechanisms proposed by \citet{kong2018water} are already the best we can hope for in their setting, in the sense that their mechanisms can work for any problem instance that can possibly be elicitable.

\subsection{Impossibility result for two agents}
Suppose we have two participants whose signals are conditionally independent, that is, the joint distribution $\mu(\omega, s_1,  s_2) = \mu(\omega) \mu(s_1|\omega) \mu(s_2|\omega)$ for all $\omega \in \Omega, s_1\in \mathcal{S}_1, s_2\in \mathcal{S}_2$. We assume that the designer is uncertain about the conditional distributions of the participants' signals $\mu(s_1|\omega), \mu(s_2|\omega)$.
We formally define the uncertainty as follows.
\begin{definition}
We say that a designer with an elicitability problem $\langle \f, M\rangle$ is minimally uncertain about the conditional distributions of the participants' signals if there exist a prior $\mu(\omega)$ and a set of possible conditional distributions $\mu(s_i|\omega)$ for each participant $i$, denoted by $M_i \subseteq \mathbb{R}^{|\mathcal{S}_i|\times|\Omega|}$, such that for any $\mu'(s_1|\omega)\in M_1, \dots, \mu'(s_n|\omega)\in M_n$, the joint distribution generated by them and the prior $\mu(\omega)$ 
$$
\mu'(\omega, s_1, \dots, s_n) = \mu(\omega)\mu'(s_1|\omega)\cdots \mu'(s_n|\omega), \ \forall \omega, s_1, \dots, s_n
$$
is a possible joint distribution $\mu'(\omega, s_1, \dots, s_n)\in M$, and each $M_i$ is open relative to the space of valid conditional distributions $(\Delta\mathcal{S}_i)^{|\Omega|}$.
\end{definition}
The main restriction here is that $M_i$ needs to be an open set relative to the space of valid conditional distributions. This means that for each participant $i$, there should exist a possible conditional distribution $\mu(s_i|\omega)$ such that the nearby conditional distributions are also possible.
Then if the designer never asks anyone for trivial reports that is a constant function of the participant's posterior, the only elicitable linear properties are the ones that are equivalent to the posteriors $\mu(\omega|s_i)$, assuming that the designer only wants scoring mechanisms for multi-task problems. More specifically, we have the following theorem.
\begin{theorem}\label{thm:linear_mult}
For a single-task/multiple-task elicitability problem $\langle \f, M\rangle$ with $n=2$ with conditional independent signals, if the designer is minimally uncertain about the conditional distributions of the participants' signals,
and for any participant $i$, the truthful report $\r_i = \G \cdot \mu(\omega|s_i)$ is a non-constant linear function of the posterior $\mu(\omega|s_i)$, then $\r$ is elicitable/elicitable by scoring mechanisms only if there is a one-to-one mapping from $\mu(\omega|s_i)$ to $\r_i$ for all $i$, i.e., 
\begin{align}\label{eqn:full_rank}
    \text{matrix } \left[\begin{array}{c} \G \\ \bm{1}^\top\end{array}\right] \text{ has rank } |\Omega|,
\end{align}
where $\bm{1}$ is the all-ones vector with length $|\Omega|$.
\end{theorem}
Here condition~\eqref{eqn:full_rank} is equivalent to that there exists a one-to-one mapping from $\mu(\omega|s_i)$ to $\r_i$ because if the matrix in~\eqref{eqn:full_rank} has rank $|\Omega|$, then it has linearly independent columns, which means that it has a left inverse $H$ such that $H\cdot\left[\begin{array}{c} \G \\ \bm{1}^\top\end{array}\right] = I.$
Then we can recover $\mu(\omega|s_i)$ from $\r_i$ as $H \cdot \left[\begin{array}{c} \r_i \\ 1\end{array}\right]= H\cdot\left[\begin{array}{c} \G \\ \bm{1}^\top\end{array}\right]\cdot \mu(\omega|s_i) = \mu(\omega|s_i)$.

According to Theorem~\ref{thm:linear_mult},\, a designer who is interested in some information that is linear in the participants' posteriors should just try to elicit their posteriors $\mu(\omega|s_i)$, if the conditions in the theorem are satisfied. The proof of the theorem is based on the necessary condition in Corollary~\ref{coro:marginal}. The proof of Theorem~\ref{thm:linear_mult} is quite involved. The high level idea is that if the columns of the matrix $ \left[\begin{array}{c} \G \\ \bm{1}^\top\end{array}\right]$ are linearly dependent, we can find two possible distribution $\mu, \mu'\in M$ so that there exist two signal realizations $s_i, s_i'$ of a participant $i$  that will lead to the same posterior in $\mu$ and $\mu'$
\begin{align*}
	\mu(\r_{-i}|s_i) = \mu'(\r_{-i}|s_i'),
\end{align*}
but the truthful reports are different
\begin{align*}
	\mathbf{G} \cdot \mu(\omega|s_i) \neq \mathbf{G} \cdot \mu'(\omega|s_i').
\end{align*}
So it violates Proposition~\ref{coro:marginal}. We defer the full proof to Appendix~\ref{app:linear_neg}.

\subsection{Prediction on the State}
In this section, we study the report function that is just the participants' posteriors about $\omega$ after observing their signals $\r_i  = \mu(\omega|s_i)$. We provide a necessary condition for $\mu(\omega|s_i)$ to be elicitable when the participants' signals are conditionally independent and the designer is uncertain about the conditional distributions of the participants' signals. 

For each participant $i$, our necessary condition will look at  the conditional distribution of the other participants' signals $\mu(\s_{-i}|\omega)$.
For every $\mu \in M$, define  $\P^\mu_i$ to be a $|\mathcal{S}_{-i}|\times |\Omega|$ matrix with entry in row $\s_{-i}$ and column $\omega$ equal to  $\mu(\s_{-i}|\omega)$. 
\begin{theorem} \label{thm:posterior_nec}
For a single-task/multiple-task elicitability problem $\langle \f, M\rangle$ with $\r_i = \mu(\omega|s_i)$ with conditional independent signals, if the designer is minimally uncertain about the conditional distributions of the participants' signals, then the condition 
$$
\text{rank}(\P^\mu_i) = |\Omega|, \ \forall i, \forall \mu
$$
is necessary for $\langle \f, M \rangle$ to be elicitable/elicitable by scoring mechanisms.
\end{theorem}
The proof can be found in Appendix~\ref{app:pos_nec}.

This necessary condition basically means that to guarantee that truthfully reporting is a strictly optimal strategy for participant $i$ (at the equilibrium), the other participants' signals need to be sufficiently correlated with the state $\omega$. Since \citet{kong2018water} proposed a mechanism that elicits $\mu(\omega|s_i)$ when the designer knows the prior $\mu(\omega)$ (see details in Section~\ref{sec:prelim_posterior}), the condition in Theorem~\ref{thm:posterior_nec} is also sufficient when the prior $\mu(\omega)$ is known.
\begin{corollary}
For a single-task/multiple-task elicitability problem $\langle \f, M\rangle$ with $\r_i = \mu(\omega|s_i)$ with known prior $\mu(\omega)$ with conditionally independent signals, if the designer is minimally uncertain about the conditional distributions of the participants' signals, then the condition that $\text{rank}(\P^\mu_i) = |\Omega|$ for all $i$ and $\mu$ is both necessary and sufficient for $\r_i = \mu(\omega|s_i)$ to be elicitable/elicitable by scoring mechanisms.
\end{corollary}
The necessity of the condition  implies that  the mechanisms proposed by \citet{kong2018water} are able to solve all the problem instances that are elicitable (by scoring mechanisms in the multi-task setting) if we consider the problem of eliciting $\mu(\omega|s_i)$ when the prior $\mu(\omega)$ is known.

\begin{corollary}
	For a single-task/multiple-task elicitability problem $\langle \f, M\rangle$ with $\r_i = \mu(\omega|s_i)$ with known prior $\mu(\omega)$ with conditionally independent signals, if the designer is minimally uncertain about the conditional distributions of the participants' signals, then all the problem instances $\langle \mathbf{f}, M\rangle$ that are elicitable/elicitable by scoring mechanisms can be solved by the mechanism proposed by \citet{kong2018water}, i.e., their mechanism will be strictly truthful for all $\mu \in M$.
\end{corollary}

\section{Discussion}
We study the elicitability of multi-task peer prediction problems. Our main contribution includes (1) we characterize the elicitable multi-task peer prediction problems when the designer only uses scoring mechanisms, (2) we are the first to study the elicitability of properties that are linear in the participants' posteriors. 
We believe that the most intriguing future direction is to further simplify our characterization and find more applications of this result, by either considering more specific settings or adopting more advanced tools. For example, our characterization does not impose any restriction on the designer's knowledge about the distribution: the set $M$ can be an arbitrary set of possible distributions. An immediate question is: can we simplify the characterization if $M$ has a certain structure? Our negative result for linear properties only used a simplified version of the necessary condition. We believe that stronger results can be proved if we deploy more of the structure of power diagrams, e.g. the convexity of the cells. Finally, our result shows that it is possible to have a simpler characterization by restricting the class of mechanisms. It may be possible to simplify our result by considering other classes of mechanisms.

\section*{Acknowledgements}
	The authors would like to thank all the anonymous reviewers for their careful reading, valuable comments, and constructive remarks.
This work is supported by the National Science Foundation under Grant No. IIS 2007887.

\newpage
\bibliographystyle{abbrvnat}
\bibliography{ref}

\begin{thebibliography}{25}
\providecommand{\natexlab}[1]{#1}
\providecommand{\url}[1]{\texttt{#1}}
\expandafter\ifx\csname urlstyle\endcsname\relax
  \providecommand{\doi}[1]{doi: #1}\else
  \providecommand{\doi}{doi: \begingroup \urlstyle{rm}\Url}\fi

\bibitem[Abernethy and Frongillo(2012)]{abernethy2012characterization}
J.~D. Abernethy and R.~M. Frongillo.
\newblock A characterization of scoring rules for linear properties.
\newblock In \emph{Proceedings of the 25th Annual Conference on Learning
  Theory}, volume~23, pages 27.1--27.13, 2012.

\bibitem[Chen et~al.(2020)Chen, Shen, and Zheng]{chen2020truthful}
Y.~Chen, Y.~Shen, and S.~Zheng.
\newblock Truthful data acquisition via peer prediction, 2020.

\bibitem[Dasgupta and Ghosh(2013)]{dasgupta2013crowdsourced}
A.~Dasgupta and A.~Ghosh.
\newblock Crowdsourced judgement elicitation with endogenous proficiency.
\newblock In \emph{Proceedings of the 22nd international conference on World
  Wide Web}, pages 319--330, 2013.

\bibitem[Faltings and Radanovic(2017)]{faltings2017game}
B.~Faltings and G.~Radanovic.
\newblock Game theory for data science: Eliciting truthful information.
\newblock \emph{Synthesis Lectures on Artificial Intelligence and Machine
  Learning}, 11\penalty0 (2):\penalty0 1--151, 2017.

\bibitem[Frongillo and Witkowski(2017)]{frongillo2017geometric}
R.~Frongillo and J.~Witkowski.
\newblock A geometric perspective on minimal peer prediction.
\newblock \emph{ACM Transactions on Economics and Computation (TEAC)},
  5\penalty0 (3):\penalty0 1--27, 2017.

\bibitem[Frongillo(2013)]{frongillo2013eliciting}
R.~M. Frongillo.
\newblock \emph{Eliciting private information from selfish agents}.
\newblock PhD thesis, UC Berkeley, 2013.

\bibitem[Jurca and Faltings(2008)]{jurca2008incentives}
R.~Jurca and B.~Faltings.
\newblock Incentives for expressing opinions in online polls.
\newblock In \emph{Proceedings of the 9th ACM Conference on Electronic
  Commerce}, pages 119--128, 2008.

\bibitem[Kong(2020)]{kong2020dominantly}
Y.~Kong.
\newblock Dominantly truthful multi-task peer prediction with a constant number
  of tasks.
\newblock In \emph{Proceedings of the Fourteenth Annual ACM-SIAM Symposium on
  Discrete Algorithms}, pages 2398--2411. SIAM, 2020.

\bibitem[Kong(2021)]{kong2021counting}
Y.~Kong.
\newblock Counting the number of people that are less clever than you.
\newblock \emph{arXiv preprint arXiv:2103.02214}, 2021.

\bibitem[Kong and Schoenebeck(2018{\natexlab{a}})]{kong2016equilibrium}
Y.~Kong and G.~Schoenebeck.
\newblock Equilibrium selection in information elicitation without verification
  via information monotonicity.
\newblock In \emph{9th Innovations in Theoretical Computer Science Conference},
  2018{\natexlab{a}}.

\bibitem[Kong and Schoenebeck(2018{\natexlab{b}})]{kong2018water}
Y.~Kong and G.~Schoenebeck.
\newblock Water from two rocks: Maximizing the mutual information.
\newblock In \emph{Proceedings of the 2018 ACM Conference on Economics and
  Computation}, pages 177--194, 2018{\natexlab{b}}.

\bibitem[Kong and Schoenebeck(2019)]{kong2019information}
Y.~Kong and G.~Schoenebeck.
\newblock An information theoretic framework for designing information
  elicitation mechanisms that reward truth-telling.
\newblock \emph{ACM Transactions on Economics and Computation (TEAC)},
  7\penalty0 (1):\penalty0 2, 2019.

\bibitem[Lambert and Shoham(2009)]{lambert2009eliciting}
N.~Lambert and Y.~Shoham.
\newblock Eliciting truthful answers to multiple-choice questions.
\newblock In \emph{Proceedings of the 10th ACM conference on Electronic
  commerce}, pages 109--118, 2009.

\bibitem[Lambert et~al.(2008)Lambert, Pennock, and
  Shoham]{lambert2008eliciting}
N.~S. Lambert, D.~M. Pennock, and Y.~Shoham.
\newblock Eliciting properties of probability distributions.
\newblock In \emph{Proceedings of the 9th ACM Conference on Electronic
  Commerce}, EC '08, pages 129--138, New York, NY, USA, 2008. Association for
  Computing Machinery.
\newblock ISBN 9781605581699.
\newblock \doi{10.1145/1386790.1386813}.

\bibitem[Liu et~al.(2020)Liu, Wang, and Chen]{liu2020surrogate}
Y.~Liu, J.~Wang, and Y.~Chen.
\newblock Surrogate scoring rules.
\newblock In \emph{Proceedings of the 21st ACM Conference on Economics and
  Computation}, pages 853--871, 2020.

\bibitem[Miller et~al.(2005)Miller, Resnick, and Zeckhauser]{MRZ05}
N.~Miller, P.~Resnick, and R.~Zeckhauser.
\newblock Eliciting informative feedback: The peer-prediction method.
\newblock \emph{Management Science}, pages 1359--1373, 2005.

\bibitem[Prelec(2004)]{prelec2004bayesian}
D.~Prelec.
\newblock A {B}ayesian {T}ruth {S}erum for subjective data.
\newblock \emph{Science}, 306\penalty0 (5695):\penalty0 462--466, 2004.

\bibitem[Radanovic and Faltings(2013)]{radanovic2013robust}
G.~Radanovic and B.~Faltings.
\newblock A robust bayesian truth serum for non-binary signals.
\newblock In \emph{Proceedings of the 27th AAAI Conference on Artificial
  Intelligence (AAAI" 13)}, number EPFL-CONF-197486, pages 833--839, 2013.

\bibitem[Radanovic and Faltings(2014)]{radanovic2014incentives}
G.~Radanovic and B.~Faltings.
\newblock Incentives for truthful information elicitation of continuous
  signals.
\newblock In \emph{Proceedings of the 28th AAAI Conference on Artificial
  Intelligence (AAAI" 14)}, number EPFL-CONF-215878, pages 770--776, 2014.

\bibitem[Schoenebeck and Yu(2020{\natexlab{a}})]{schoenebeck2020learning}
G.~Schoenebeck and F.-Y. Yu.
\newblock Learning and strongly truthful multi-task peer prediction: A
  variational approach, 2020{\natexlab{a}}.

\bibitem[Schoenebeck and Yu(2020{\natexlab{b}})]{schoenebeck2020two}
G.~Schoenebeck and F.-Y. Yu.
\newblock Two strongly truthful mechanisms for three heterogeneous agents
  answering one question.
\newblock In \emph{International Conference on Web and Internet Economics}.
  Springer, 2020{\natexlab{b}}.

\bibitem[Shnayder et~al.(2016)Shnayder, Agarwal, Frongillo, and
  Parkes]{shnayder2016informed}
V.~Shnayder, A.~Agarwal, R.~Frongillo, and D.~C. Parkes.
\newblock Informed truthfulness in multi-task peer prediction.
\newblock In \emph{Proceedings of the 2016 ACM Conference on Economics and
  Computation}, pages 179--196, 2016.

\bibitem[Witkowski and Parkes(2013)]{Witkowski2013LearningTP}
J.~Witkowski and D.~Parkes.
\newblock Learning the prior in minimal peer prediction.
\newblock 2013.

\bibitem[Witkowski and Parkes(2012)]{witkowski2012peer}
J.~Witkowski and D.~C. Parkes.
\newblock Peer prediction without a common prior.
\newblock In B.~Faltings, K.~Leyton{-}Brown, and P.~Ipeirotis, editors,
  \emph{Proceedings of the 13th {ACM} Conference on Electronic Commerce, {EC}
  2012, Valencia, Spain, June 4-8, 2012}, pages 964--981. {ACM}, 2012.
\newblock \doi{10.1145/2229012.2229085}.
\newblock URL \url{https://doi.org/10.1145/2229012.2229085}.

\bibitem[Zhang and Chen(2014)]{zhang2014elicitability}
P.~Zhang and Y.~Chen.
\newblock Elicitability and knowledge-free elicitation with peer prediction.
\newblock In \emph{Proceedings of the 2014 international conference on
  Autonomous agents and multi-agent systems}, pages 245--252, 2014.

\end{thebibliography}

\newpage
\appendix
\section{Table of notations}

\begin{tabular}{p{0.2\textwidth}p{0.35\textwidth}p{0.35\textwidth}}
\toprule
symbols & single tasks & multiple tasks\\
\midrule
  space of signal & $\bm{\mathcal{S}} := \times_i \mathcal{S}_i$, and $\mathbf{s}\in \bm{\mathcal{S}}$ & $\bm{\mathcal{S}}:=\times_i \mathcal{S}_i$, and $\mathbf{s}^{(1:T)}\in \bm{\mathcal{S}}^T$\\ \midrule
  common prior & $\mu$ a distribution on $\Omega\times \mathcal{S}_1\dots \times \mathcal{S}_n$ & $\mu^{\otimes T}$, each task is identically and independently sampled from $\mu$\\ \midrule
  Designer's knowledge & $M$ a subset of distributions that contains the underlying distribution $\mu\in M$& same\\ \midrule
    report space & $\bm{\mathcal{R}}:=\times_i\mathcal{R}_i$ 
    & \text{same} 
    \\ \midrule
  private signal  & $\mathbf{S}:=(S_i)_i$ the random variable of signals sampled from $\mu$ & $\mathbf{S}^{(1:T)}:=(\mathbf{S}^{(1:T)}_i)_i = (S_i^{(t)})_{i,t}$ a random variable sampled from $\mu^{\otimes T}$\\ \midrule
  truthful map & $\mathbf{f} = (f_{i,\mu})_{i,\mu}$ & same\\ \midrule
  truthful report &  $\mathbf{R} = (R_i)_i$, a random variable $(f_{i,\mu}(S_i))_{i}$ where $(S_i)_i\sim \mu$ & $\mathbf{R}^{(1:T)} = (R^{(1:T)}_i)_{i}= (R^{(t)}_i)_{i,t}$, a random variable\\ \midrule
  strategy &  $\bm{\sigma} = (\sigma_i)_i$, a collection of $n$ independent random mappings where $\sigma_i: \mathcal{S}_i\to \Delta(\mathcal{R}_i)$. &  $\bm{\sigma} = (\sigma_i^{(1:T)})_{i} = (\sigma_i^{(t)})_{i,t}$ a collection of $nT$ independent random mapping $\sigma_i^{(t)}:\mathcal{S}_i\to \Delta(\mathcal{R}_i)$\\ \midrule
  strategic report &  $\tilde{\mathbf{R}} = (\tilde{R_i})_i = (\sigma_i(S_i))_i$, random variables $(\sigma_{i}(S_i))_{i}$ depends on $(\sigma_i)_i$ and $(S_i)_i$  & $\tilde{\mathbf{R}}^{(1:T)} = (\tilde{R}^{(1:T)}_i)_{i} = (\tilde{R}^{(t)}_i)_{i,t} = (\sigma_i^{(t)}(S_i^{(t)}))_{i,t}$\\ \midrule
  payment & $\mathbf{p} = (p_i)_i$ where $p_i:\times_i \mathcal{R}_i\to \mathbb{R}$ & $\mathbf{p} = (p_i)_i$ where $p_i:\times_i \mathcal{R}_i^T\to \mathbb{R}$ \\
\bottomrule
\end{tabular}

\section{Missing Proofs in Section~\ref{sec:pre}} \label{app:pre_ball}

Given a distribution $\lambda$ on a finite set $\Omega\times \mathcal{S}_1\times\dots\times\mathcal{S}_n$ and $\epsilon>0$, let $B(\lambda, \epsilon) = \{\mu: \|\mu-\lambda\|_1\le \epsilon\}\subset \Delta(\Omega\times \mathcal{S}_1\times\dots\times\mathcal{S}_n)\}$ be the collection of distributions whose total variation distance from $\lambda$ is smaller than $\epsilon$.  Note that if $\epsilon = 0$, $B(\lambda, 0)$ is a singleton and we know exactly the true distribution.  There are several truthful mechanisms, e.g., \citet{MRZ05}.  The following theorem shows we can still have truthful mechanisms when $\epsilon>0$ is small enough.
\begin{theorem}\label{thm:ball}
If $f_i(s_i) = s_i$ for all $i$ and $\lambda\in \Delta(\Omega\times \mathcal{S}_1, \dots \mathcal{S}_n)$ has full support $\lambda>0$ and is stochastic relevant so that $\lambda(\mathbf{s}_{-i}|s_i)\neq \lambda(\mathbf{s}_{-i}|s_i')$ for all $i\in [n]$ and distinct $s_i$ and $s_i'$, there exists $\epsilon>0$ such that $\langle \f, B(\lambda, \epsilon)\rangle$ is elicitable.
\end{theorem}
The idea is very similar to the maximal robust mechanisms in \citet{frongillo2017geometric}.  However, instead of the joint distribution being close to a center $\lambda$, their result requires all conditional distributions are close to a center.  The proof is straightforward and it is in the appendix.

\begin{proof}[Proof of Corollary~\ref{thm:necessary_1}]
If $Q_i(r_i)\cap Q_i(r_i') \neq \emptyset$, then there exist $\mu_1$ and $\mu_2$ with
$$
\mu_1(\r_{-i}|s_i) = \mu_2(\r_{-i}|s_i'), \quad f(s_i) = r_i, \quad f(s_i') = r_i'.
$$
Then if there exists a truthful payment rule $\p(\r)$, then by the definition of elicitability Definition~\ref{def:elicitable} it should satisfy that
$$
\E_{\r_{-i} \sim \mu_1(\r_{-i} | s_i)} [p_i(r_i, \r_{-i})] > \E_{\r_{-i} \sim \mu_1(\r_{-i} | s_i)} [p_i(r_i', \r_{-i})],
$$
$$
\E_{\r_{-i} \sim \mu_2(\r_{-i} | s_i')} [p_i(r_i', \r_{-i})] > \E_{\r_{-i} \sim \mu_2(\r_{-i} | s_i')} [p_i(r_i, \r_{-i})].
$$
which is impossible if $\mu_1(\r_{-i}|s_i) = \mu_2(\r_{-i}|s_i')$.
\end{proof}

\begin{proof}[Proof of Theorem~\ref{thm:ball}]
Because $\lambda$ is stochastic relevant, $\langle \f, \{\lambda\}\rangle$ is elicitable.~\cite{MRZ05}

By Theorem~\ref{thm:single_chr}, for agent $i$, there exists a power diagram with sites $\mathbf{v}^r\in \mathbb{R}^{|\mathcal{R}_{-i}|}$ and weights $w^r\in \mathbb{R}$ for all $r\in \mathcal{R}_i$ such that for all $r_i\in \mathcal{R}_i$, 
$r_i = \arg\min_s\{\langle \lambda(\mathbf{r}_{-i}|r_i), \mathbf{v}^s\rangle-w^s\}$, and for any distinct pair $r_i, r_i'$
$$\langle \lambda(\mathbf{r}_{-i}|r_i), \mathbf{v}^{r_i}\rangle-w^{r_i}>\langle \lambda(\mathbf{r}_{-i}|r_i), \mathbf{v}^{r_i'}\rangle-w^{r_i'}$$

Now we want to prove the same power diagram works for $\langle \f, B(\lambda, \epsilon)\rangle$ when $\epsilon$ is small enough: For all $r_i\in \mathcal{R}_i$ and $\mu\in B(\lambda, \epsilon)$, $cell(\mathbf{v}^{r_i})$ contains $Q_i(r_i)$.

First we can bound the distance between conditional distributions by the distance between their joint distributions.  For any $\mu>0$ and $r_i\in \mathcal{R}_i$
\begin{align*}
    \|\mu(\mathbf{r}_{-i}|r_i)-\lambda(\mathbf{r}_{-i}|r_i)\|_1 =& \sum_{\mathbf{r}_{-i}} \left|\frac{\mu(r_i, \mathbf{r}_{-i})}{\mu(r_i)}-\frac{\lambda(r_i, \mathbf{r}_{-i})}{\lambda(r_i)}\right|\\
    \le& \sum_{\mathbf{r}_{-i}} \left|\frac{\mu(r_i, \mathbf{r}_{-i})}{\mu(r_i)}-\frac{\mu(r_i, \mathbf{r}_{-i})}{\lambda(r_i)}\right|+\sum_{\mathbf{r}_{-i}}\left|\frac{\mu(r_i, \mathbf{r}_{-i})}{\lambda(r_i)}-\frac{\lambda(r_i, \mathbf{r}_{-i})}{\lambda(r_i)}\right|\\
    \le& |\mu(r_i)^{-1}-\lambda(r_i)^{-1}|\sum_{\mathbf{r}_{-i}} \left|\mu(r_i, \mathbf{r}_{-i})\right|+\lambda(r_i)^{-1}\sum_{\mathbf{r}_{-i}}\left|\mu(r_i, \mathbf{r}_{-i})-\lambda(r_i, \mathbf{r}_{-i})\right|\\
    \le& |\mu(r_i)^{-1}-\lambda(r_i)^{-1}|+\lambda(r_i)^{-1}\|\mu-\lambda\|_1
\end{align*}
Now it is sufficient to bound the first term by $\|\mu-\lambda\|_1$.  If $\|\mu-\lambda\|_1<\epsilon$, we can take $\epsilon$ smaller than $\min_{r_i}\lambda(r_i)/2$.  Thus, for all $r_i$, we have $|\mu(r_i)-\lambda(r_i)|<\lambda(r_i)/2$, and $|\mu(r_i)^{-1}-\lambda(r_i)^{-1}|= |\mu(r_i)-\lambda(r_i)|/(\mu(r_i)\lambda(r_i))\le \frac{2}{\lambda(r_i)^2}\|\mu-\lambda\|_1$.  Therefore, we have
\begin{equation}\label{eq:ball1}
    \|\mu(\mathbf{r}_{-i}|r_i)-\lambda(\mathbf{r}_{-i}|r_i)\|_1\le \left(\frac{2}{\lambda(r_i)^2}+\frac{1}{\lambda(r_i)}\right)\|\mu-\lambda\|_1\le  \frac{3}{\lambda(r_i)^2}\|\mu-\lambda\|_1
\end{equation}
On the other hand, since $\mathcal{R}_i$ is a finite set and $\lambda>0$, we can pick $\epsilon_0>0$ small enough such that for all $r_i\neq r_i'$
\begin{equation}\label{eq:ball2}
    3\epsilon_0+\langle \lambda(\mathbf{r}_{-i}|r_i), \mathbf{v}^{r_i}\rangle-w^{r_i}<\langle \lambda(\mathbf{r}_{-i}|r_i), \mathbf{v}^{r_i'}\rangle-w^{r_i'}.
\end{equation}
Now for all distinct $r_i, r_i'$ and $\mu\in B(\lambda, \min_{r_i}\frac{\lambda({r_i})^2}{ \|\mathbf{v}^{r_i}\|_\infty}\epsilon_0)$, we have
\begin{align*}
    &\left(\langle \lambda(\mathbf{r}_{-i}|r_i), \mathbf{v}^{r_i'}\rangle-w^{r_i'}\right)-\left(\langle \lambda(\mathbf{r}_{-i}|r_i), \mathbf{v}^{r_i}\rangle-w^{r_i}\right)\\
    =& \langle \lambda(\mathbf{r}_{-i}|r_i), \mathbf{v}^{r_i'}-\mathbf{v}^{r_i}\rangle-w^{r_i'}+w^{r_i}\\
    \ge& \langle \lambda(\mathbf{r}_{-i}|r_i)-\mu(\mathbf{r}_{-i}|r_i), \mathbf{v}^{r_i'}-\mathbf{v}^{r_i}\rangle+3\epsilon_0\tag{by Eqn.~\eqref{eq:ball2}}\\
    \ge& 3\epsilon_0-\|\lambda(\mathbf{r}_{-i}|r_i)-\mu(\mathbf{r}_{-i}|r_i)\|_1\|\mathbf{v}^{r_i}\|_\infty\tag{Holder's inequality}\\
    >& 3\epsilon_0-\frac{3\|\mathbf{v}^{r_i}\|_\infty}{\lambda(r_i)^2}\|\lambda-\mu\|_1\tag{by Eqn.~\eqref{eq:ball1}}\\
    \ge& 0\tag{$\|\mu-\lambda\|_1\le \frac{\lambda(r_i)^2}{ \|\mathbf{v}^{r_i}\|_\infty}\epsilon_0$}
\end{align*}
This completes our proof.
\end{proof}

\section{Missing proofs in Section~\ref{sec:mult}} \label{app:mult}

\subsection{Proof of Theorem~\ref{thm:equiv}}

First it is apparent that a mechanism that guarantees strict truthfulness also guarantees strict truthfulness under consistent strategies.
In the following proposition we show if there is a scoring mechanism so that a multi-task problem $\langle \f, M\rangle$ is elicitable under consistent strategies there is a scoring mechanism such that $\langle \f, M\rangle$ is elicitable under general strategies defined in Definition~\ref{def:strategy}.
\begin{proposition}\label{prop:consistent}
If there exists a scoring rule mechanism $\mathbf{p} = \{p_i^t:i\in [n], t\in [T]\}$ so that a multi-task problem $\langle \f, M\rangle$ is elicitable under consistent strategies,  there is a scoring mechanism such that $\langle \f, M\rangle$ is elicitable under general strategies.
\end{proposition}
\begin{proof}[Proof of Proposition~\ref{prop:consistent}]
We construct such scoring mechanism against general strategies through symmetrization.  

Since $\mathbf{p}$ is a scoring mechanism, agent $i$'s payment $p_i$ can be decomposed as $p_i(\r_i^{(1:T)},\r_{-i}^{(1:T)}) = \sum_{t=1}^T p_i^{(t)}(r_i^{(t)}, \r_{-i}^{(1:T)})$ for all $\r\in \bm{\mathcal{R}}$ where for agent $i$, $p_i^{(t)}$ only depends on his report on task $i$.  We define a new payment $\hat{p}_i$ through symmetrization: Let ${\rm Sym}(T)$ be the collection of all permutations on set $[T]$.
$$\hat{p}_i(\mathbf{r}) := \frac{1}{T!}\sum_{\tau\in [T],\pi\in {\rm Sym}(T)}p_i^{(\tau)}\left(r_i^{(\pi(1))}, r_{-i}^{(\pi(1))}, r_{-i}^{(\pi(2))},\dots ,r_{-i}^{(\pi(T))}\right).$$

First it is easy to see if all agents use consistent strategies $\bm{\sigma}$, the expectation of $\hat{p}_i$ equals the expectation of $p_i$ 
\begin{equation}\label{eq:consistent3}
    \E_{\bm{\sigma}, \mathbf{S}}\left[\hat{p}_i(\tilde{\mathbf{R}}) \right]= \E_{\bm{\sigma}, \mathbf{S}}\left[p_i(\tilde{\mathbf{R}})\right].
\end{equation}
Formally, because when $\bm{\sigma}$ is consistent the distribution of $\tilde{\mathbf{R}}$ is exchangeable on tasks so that the distribution of $(\tilde{R}^{(1)}\dots \tilde{R}^{(T)})$ is identical to the distribution of $(\tilde{R}^{(\pi(1))}\dots \tilde{R}^{(\pi(T))})$ for any permutation $\pi\in {\rm Sym}(T)$, the expectation of $\hat{p}_i$ is the sum of $T!$ identical terms, and Eqn.~\eqref{eq:consistent3} holds. 

Then we show that given any agent $i$'s general strategy $\sigma_i = (\sigma_i^1,\dots,\sigma_i^T)$, if all other agents are truth telling, there exists a consistent 
$$\hat{\sigma}_i := \frac{1}{T}\sum_t \sigma_i^t$$
such that agent $i$'s expected payment of $\hat{p}_i$ under strategy $\sigma_i$ is equal to the expected payment of $p_i$ under $\hat{\sigma}_i$, \begin{equation}\label{eq:consistent2}
    \E_{\sigma_i, \mathbf{S}}\left[\hat{p}_i\left(\tilde{\mathbf{R}}^{(1:T)},\mathbf{R}_{-i}^{(1:T)}\right)\right] = \E_{\hat{\sigma}_i, \mathbf{S}}\left[p_i\left(\tilde{\mathbf{R}}^{(1:T)},\mathbf{R}_{-i}^{(1:T)}\right)\right]
\end{equation}
Note that if Eqn.~\eqref{eq:consistent2} holds, with Eqn.~\eqref{eq:consistent3} we completes the proof.  When every other agents are truth telling, the expected payment under any general nontruthful strategy $\sigma_i$ is strictly less than the expected payment of truth telling,
\begin{align*}
    &\E_{\sigma_i, \mathbf{S}}\left[\hat{p}_i\left(\tilde{\mathbf{R}}^{(1:T)},\mathbf{R}_{-i}^{(1:T)}\right)\right]\\
    =& \E_{\hat{\sigma}_i, \mathbf{S}}\left[p_i\left(\tilde{\mathbf{R}}^{(1:T)},\mathbf{R}_{-i}^{(1:T)}\right)\right]\tag{by \eqref{eq:consistent2}}\\
    <& \E_{\mathbf{S}}\left[p_i\left(\tilde{\mathbf{R}}\right)\right]\tag{$p_i$ is truthful under consistent strategies.}\\
    =& \E_{\mathbf{S}}\left[\hat{p}_i\left(\tilde{\mathbf{R}}\right)\right]\tag{by \eqref{eq:consistent3} and truth telling is an consistent strategy}
\end{align*}
The above inequality is strict, because if $\sigma_i$ is not truthful, the average $\hat{\sigma}_i$ is also not truthful.

Finally, let's prove Eqn.~\eqref{eq:consistent2}.  We set $\mathbf{R}_{-i}^{\otimes (T-1)}$ be a sequence of $T-1$ iid truthful report on a generic tasks, and  $(\mathbf{R}_{-i}^{(\pi(2))},\dots,\mathbf{R}_{-i}^{(\pi(T))})$ has the same distribution as $\mathbf{R}_{-i}^{\otimes (T-1)}$ for any permutation $\pi\in {\rm Sym}(T)$, because the distribution is exchangeable.  With this notion, due to the linearity of expectation, we have 
\begin{align*}
    \E_{\sigma_i, \mathbf{S}}[\hat{p}_i(\tilde{\mathbf{R}})]
    =& \frac{1}{T}\E\left[\sum_{\tau, \tau'}p_i^{(\tau)}\left(\sigma_i^{(\tau')}(S_i^{(\tau')}),\mathbf{R}_{-i}^{(\tau')},\mathbf{R}_{-i}^{\otimes (T-1)}\right)\right]\\
    =& \frac{1}{T}\E_{\sigma_i, \mathbf{S}}\left[\sum_{\tau, \tau'}p_i^{(\tau)}\left(\sigma_i^{\tau'}(S_i'),f_{-i,\mu}(S_{-i}'),\mathbf{R}_{-i}^{\otimes (T-1)}\right)\right]\tag{Let $(S_i', S_{-i}')$ be sampled from $\mu$.}\\
    =& \frac{1}{T}\E_{ \mathbf{S}}\left[\sum_{\tau, \tau'}\sum_{r_i\in \mathcal{R}_i}\Pr\left[\sigma_i^{(\tau')}(S_i') = r_i\right]p_i^{(\tau)}\left(r_i,f_{-i,\mu}(S_{-i}'),\mathbf{R}_{-i}^{\otimes (T-1)}\right)\right]\\
    =& \E_{ \mathbf{S}}\left[\sum_{\tau}\sum_{r_i\in \mathcal{R}_i}\left(\frac{1}{T}\sum_{\tau'}\Pr\left[\sigma_i^{(\tau')}(S_i') = r_i\right]\right)p_i^{(\tau)}\left(r_i,f_{-i,\mu}(S_{-i}'),\mathbf{R}_{-i}^{\otimes (T-1)}\right)\right]\\
    =& \E_{ \mathbf{S}}\left[\sum_{\tau}\sum_{r_i\in \mathcal{R}_i}\Pr\left[\hat{\sigma}_i(S_i') = r_i\right]p_i^{(\tau)}\left(r_i,f_{-i,\mu}(S_{-i}'),\mathbf{R}_{-i}^{\otimes (T-1)}\right)\right]\tag{$\hat{\sigma}_i = \frac{1}{T}\sum_t \sigma_i^t$}\\
    =& \E_{\hat{\sigma}_i,\mathbf{S}}\left[\sum_{\tau}p_i^{(\tau)}\left(\hat{\sigma}_i(S_i'),f_{-i,\mu}(S_{-i}'),\mathbf{R}_{-i}^{\otimes (T-1)}\right)\right] = \E_{\hat{\sigma}_i,\mathbf{S}}\left[p_i(\tilde{\mathbf{R}})\right]
\end{align*}
which completes the proof.
\end{proof}

\subsection{Characterization for polynomial-size scoring mechanisms}\label{app:main_simplified}

Suppose the designer only uses scoring mechanisms with 
$$p_i^{(t)}(\tilde{r}_i^{(t)}, \tilde{\r}_{-i}^{(1:T)}) = p_i^{(t)}(\tilde{r}_i^{(t)}, \tilde{r}_{-i}^{(t)}, Y(\tilde{\r}_{-i}^{(-t)})),$$
where $Y$ is an arbitrary function with range $\mathcal{Y}$. And $\mathcal{Y}$ has a polynomial size.
 We first define  the power diagram constraint for fixed marginal distribution of $\mu(Y(\r_{-i}^{(1:T-1)}))$. The definition is similar to Definition~\ref{def:mult_nec}. To simplify the notation, we write $Y({\r}_{-i}^{(1:T-1)})$ as $Y$.
 
Define $M(\mu(Y))$ to be the set of all distributions  that has marginal distribution of $Y$ equal to $\mu(Y)$,
$$
M(\mu(Y)) = \big\{ \mu' \in M: \mu'(Y) = \mu(Y) \big\}
$$
 Also define $Q_i(r_i, \mu(Y))$ to be the set of participant $i$'s possible posteriors about the others' reports when participant $i$'s truthful report is $r_i$ and the marginal distribution of $Y$ is  $\mu(Y)$,
$$
Q_i(r_i, \mu(Y))=\{ \mu(\r_{-i}|s_i): \mu\in M(\mu(Y)),\ s_i \text{ satisfies } f_{i,\mu}(s_i) = r_i\}.
$$ 
Then  the power diagram constraint for fixed marginal distribution of $\mu(Y)$ is defined as follows.
\begin{definition}
	A multi-task problem $\langle \f, M\rangle$ satisfies the \emph{power diagram constraint for fixed marginal distributions of $\mu(Y)$} if for all $i$ and $\mu(Y)$, $\{ Q_i(r_i, \mu(Y))\}_{r_i \in \mathcal{R}_i}$ can be fitted into a power diagram, which means that there exists a power diagram in dimension $|\mathcal{R}_{-i}|$ with $|\mathcal{R}_i|$ cells defined by $\{\mathbf{v}^{r_i}\in \mathbb{R}^{|\mathcal{R}_{-i}|}: r_i\in \mathcal{R}_i\}$ with associated weights $\{w^{r_i}\in \mathbb{R}:r_i\in \mathcal{R}_i\}$, such that each $Q_i(r_i, \mu(Y))$ falls into a distinct cell,
	$$
Q_i(r_i, \mu(Y)) \subseteq \text{cell}(\mathbf{v}^{r_i}), \text{ for all } r_i \in \mathcal{R}_i.  
	$$
\end{definition}
Then we have the following theorem.
\begin{theorem}
If the designer only uses scoring mechanisms with 
$$p_i^{(t)}(\tilde{r}_i^{(t)}, \tilde{\r}_{-i}^{(1:T)}) = p_i^{(t)}(\tilde{r}_i^{(t)}, \tilde{r}_{-i}^{(t)}, Y(\tilde{\r}_{-i}^{(-t)})),$$
where $Y$ is an arbitrary function with range $\mathcal{Y}$, then a multi-task problem is elicitable if and only if (1) it satisfies the power diagram constraint for fixed marginal distribution of $\mu(Y)$; (2) the site $\mathbf{v}^{r_i}$ for different $\mu(Y)$ is an affine function of $\mu(Y)$ and the weights $w^{r_i}$ for different $\mu(Y)$ is also an affine function of $\mu(Y)$. 
Formally, for every $i\in[n]$ and every $r_i \in \mathcal{R}_i$, there exist a  matrix $\mathbf{D}_{r_i}$ with $|\mathcal{R}_{-i}|$ rows and $|\mathcal{Y}|$ columns and a vector $\mathbf{e}_{r_i}\in \mathbb{R}^{|\mathcal{R}_{-i}|}$ with
	$$
	\mathbf{v}^{r_i}(\mu(Y)) = \mathbf{D}_{r_i} \cdot \mu(Y) + \mathbf{e}_{r_i},
	$$
	and there exists a vector $\mathbf{h}_{r_i}$ with length $|\mathcal{Y}|$ such that 
	$$
	w^{r_i}(\mu(Y)) = \mathbf{h}_{r_i}^\top \cdot \mu(Y).
	$$
	Moreover, if such power diagrams and affine functions exist, we can find a mechanism that is strictly truthful for any $\mu\in M$ with payments defined by the entries of  $\, \mathbf{D}, \mathbf{e}$ and $\mathbf{h}$ as
	\begin{align*} 
p_i^{(t)}(r_i,  \r_{-i}^{(t)}, Y(\tilde{\r}_{-i}^{(-t)})) = -\mathbf{D}_{r_i}[\r_{-i}^{(t)}, Y(\tilde{\r}_{-i}^{(-t)})] \ - \mathbf{e}_{r_i}[\r_{-i}^{(t)}] +\mathbf{h}_{r_i}[Y(\tilde{\r}_{-i}^{(-t)})], \quad \forall r_i \in \mathcal{R}_i,\ \r_{-i}^{(t)},\ \r_{-i}^{(-t)}.
\end{align*}
\end{theorem}
The proof of the theorem is entirely similar to the proof of Theorem~\ref{thm:mult_suff}.

\section{Missing proofs in Section~\ref{sec:linear}}

\subsection{Proof of Theorem~\ref{thm:linear_mult}} \label{app:linear_neg}


Before starting the main proof, we give a useful lemma about distributions with conditional independent signals.
\begin{lemma} \label{lem:pos_cond}
	If a distribution $\mu(\omega, s_1, s_2)$ has conditional independent signals, i.e., $\mu(\omega, s_1, s_2) = \mu(\omega) \mu(s_1|\omega) \mu(s_2|\omega)$, then participant $1$'s posterior about $s_2$ satisfies
	$$
	\mu(s_2|s_1) = \sum_{\omega \in \Omega} \mu(s_2|\omega) \mu(\omega|s_1).
	$$
\end{lemma}
\begin{proof}
	It is because $$\mu(s_2|s_1) = \sum_\omega \mu(s_2, \omega|s_1) = \sum_\omega \mu(s_2|\omega, s_1)\mu(\omega|s_1) = \sum_\omega \mu(s_2|\omega) \mu(\omega|s_1).$$
	Here in the last equation, we've used the conditional independence: $\mu(s_2|\omega, s_1)= \frac{\mu(\omega, s_1, s_2)}{\mu(\omega, s_1)} = \frac{\mu(\omega)\mu(s_1|\omega) \mu(s_2|\omega)}{\mu(\omega)\mu(s_1|\omega)} =\mu(s_2|\omega).$
\end{proof}

Consider a multi-task elicitability problem  $\langle \f, M\rangle$ with two agents and with conditional independent signals with $\r_i = \G \cdot \bm{\mu}(\omega|s_i)$, and the designer is minimally uncertain about the conditional distributions of the participants' signals.
 We show that if rank$([\G^\top, \bm{1}])<|\Omega|$, there exists a pair of  distributions $\mu^*, \mu'\in M$ that will violate the condition in Proposition~\ref{coro:marginal}.

We first choose a distribution $\mu^*\in M$ as follows. 
Since the designer is minimally uncertain about the conditional distributions of the participants' signals, there exists a prior $\mu(\omega)$ and a set of possible conditional distributions for each participant $M_i \subseteq \mathbb{R}^{|\mathcal{S}_i|\times|\Omega|}$, such that the joint distributions generated by them are all possible, and each $M_i$ is open relative to the space of valid conditional distributions $(\Delta\mathcal{S}_i)^{|\Omega|} $.
Suppose that the first agent's set of possible likelihood functions $M_1 \subseteq \mathbb{R}^{|\mathcal{S}_1|\times |\Omega|}$ contains a ball with radius $\varepsilon$ centered at $\mu^*(s_1|\omega)$, and $M_2$ contains a ball centered at $\mu^*(s_2|\omega)$.\footnote{Here $M_i$ contains a ball means that there is a ball $\mathcal{B}\subseteq \mathbb{R}^{|\mathcal{S}_i|\times|\Omega|}$ such that the intersection of the ball and the space of valid conditional distributions is still in $M_i$, i.e.,  $\mathcal{B} \cap (\Delta\mathcal{S}_i)^{|\Omega|} \subseteq M_i$. } 
Let $\mu^*(\omega, s_1, s_2)$ be the joint distribution generated by $\mu(\omega)$, $\mu^*(s_1|\omega)$ and $\mu^*(s_2|\omega)$.
Then for a certain $s_1^*\in \mathcal{S}_1$,
 participant $1$'s possible posterior about $\omega$ after seeing $s_1^*$
 must contain a ball centered at 
$$
\mu^*(\omega|s_1^*) \propto \mu(\omega) \mu^*(s_1^*|\omega)
$$
with some radius $\varepsilon^*>0$. 

We then find another distribution $\mu'\in M$ such that 
\begin{enumerate}[label=(\alph*)]
\item the marginal distribution of $\r_2$ remains the same, $\mu'(\r_2) = \mu^*(\r_2)$;
\item participant $1$'s truthful report after seeing $s_1^*$ is different, $\G\cdot \bm{\mu}'(\omega|s_1^*) \neq \G\cdot \bm{\mu}^*(\omega|s_1^*)$;
\item but participant $1$'s posterior about $\r_2$ after seeing $s_1^*$ is unchanged, $\mu'(\r_2|s_1^*) = \mu^*(\r_2|s_1^*).$
\end{enumerate}
Then according to Proposition~\ref{coro:marginal}, $\r$ is not elicitable.

Since the report function $\r_i = \G \cdot \bm{\mu}(\omega|s_i)$ is nonconstant, there exists a length-$|\Omega|$ vector $\bm{\beta}$  with $\G\cdot \bm{\beta} \neq 0$ and $\bm{1}^\top\cdot \bm{\beta} = 0$. Since the possible posterior of the first agent after seeing $s_1^*$ contains a ball centered at $\mu^*(\omega|s_1^*)$, there exists a small enough $\delta$ such that $\mu'(\omega|s_1^*) = \mu^*(\omega|s_1^*) + \delta \cdot \bm{\beta}$ is a possible posterior of the first agent and condition (b) is satisfied, i.e.,
\begin{align} \label{eqn:linear_eq1}
\G\cdot \bm{\mu}^*(\omega|s_1^*) \neq \G\cdot \bm{\mu}'(\omega|s_1^*).
\end{align}
We will use this vector $\bm{\beta}$ to construct a likelihood function $\mu'(s_2|\omega)\in M_2$ so that (a) and (c) are also satisfied.

Let $\Q^*$ be an $|\mathcal{S}_2|\times |\Omega|$ matrix with $\mu^*(s_2|\omega)$ at row $s_2$ and column $\omega$.
If rank$([\G^\top, \bm{1}])<|\Omega|$, then there exists a length-$|\Omega|$ vector $\bm{\alpha}$ such that $\bm{\alpha} \neq 0$ 
\begin{align*}
    \left[\begin{array}{c} \G \\ \bm{1}^\top\end{array}\right] \text{diag}(\bm{\mu}(\omega))\cdot \bm{\alpha} = 0,
\end{align*}  
where $\text{diag}(\bm{\mu}(\omega))$ represents the $|\Omega|\times |\Omega|$ matrix with the prior $\mu(\omega)$ on the diagonal.
Define matrix $$\mathbf{Q}' = \Q^* + \mathbf{k} \cdot \bm{\alpha}^\top$$ with
$$
\mathbf{k} =  \frac{- \delta\Q^*  \bm{\beta}}{\bm{\alpha}^\top\cdot \bm{\mu}'(\omega|s_1^*)}.
$$
(Here it is WLOG to assume that the denominator is non-zero. The reason is as follows.
WLOG we can assume that $\bm{\alpha}^\top\cdot \bm{\mu}^*(\omega|s_1^*) \neq 0$ because if it is equal to zero, we can shift $\bm{\mu}^*(\omega|s_1^*)$ within the small ball so that it becomes non-zero. Then there must exist small enough $\delta$ such that the denominator
$\bm{\alpha}^\top\cdot \bm{\mu}'(\omega|s_1^*)$ is non-zero.)
It is easy to verify that the vector $\mathbf{k}$ we picked satisfies $\bm{1}^\top \cdot \mathbf{k} = 0$ because $\bm{1}^\top \Q^* = \bm{1}^\top$ and $\bm{1}^\top \bm{\beta} = 0$. Therefore by choosing small enough $\delta$, matrix $\Q' = \Q^* + \mathbf{k} \cdot \bm{\alpha}^\top$ will represent a valid likelihood function $\mu'(s_2|\omega)\in M_2$, because $M_2$ contains a ball with radius $\varepsilon$ centered at $\Q^*$. Let $\mu'(\omega, s_1, s_2)$ be the joint distribution generated by $\mu(\omega), \mu'(s_2|\omega)$ and any $\mu(s_1|\omega)\in M_1$ that will give the desired $\mu'(\omega|s_1^*) = \mu^*(\omega|s_1^*) + \delta \cdot \bm{\beta}$.

We first show that condition (c) is satisfied.
By our construction of $\mathbf{k}$, participant 1's posteriors about participant $2$'s signal $\mu'(s_2|s_1^*)$ will remain the same as $\mu^*(s_2|s_1^*)$. Because by Lemma~\ref{lem:pos_cond}, 
\begin{align*}
   \bm{\mu}'(s_2|s_1^*) = \Q'\cdot \bm{\mu}'(\omega|s_1^*) = & \ (\Q^* + \mathbf{k}\bm{\alpha}^\top) \bm{\mu}'(\omega|s_1^*)\\
    = &\ \Q^* \cdot \bm{\mu}'(\omega|s_1^*) + \mathbf{k}\big(\bm{\alpha}^\top \bm{\mu}'(\omega|s_1^*)\big)\\
    =& \ \Q^* \cdot (\bm{\mu}^*(\omega|s_1^*) +\delta \cdot \bm{\beta}) - \delta \Q^*  \bm{\beta}\\
    = & \ \Q^* \cdot \bm{\mu}^*(\omega|s_1^*)\\
    = & \ \bm{\mu}^*(s_2|s_1^*).
\end{align*}
Notice that this is not equivalent to (c) because we need the posterior about participant $2$'s report $\r_2$ to be unchanged. 
However, it suffices to prove that for all $s_2 \in \mathcal{S}_2$, the  posteriors $\mu^*(\omega|s_2)$ and $\mu'(\omega|s_2)$ lead to the same report $\G\cdot \bm{\mu}^*(\omega|s_2) = \G \cdot \bm{\mu}'(\omega|s_2)$.
By our construction of $\Q'$
$$
\Q' = \Q^* + \mathbf{k} \cdot \bm{\alpha}^\top
$$
and our selection of $\bm{\alpha}$ which guarantees
\begin{align*}
    \left[\begin{array}{c} \G \\ \bm{1}^\top\end{array}\right] \text{diag}(\bm{\mu}(\omega))\cdot \bm{\alpha} = 0,
\end{align*}  
it holds that
\begin{align*}
    \left[\begin{array}{c} \G \\ \bm{1}^\top\end{array}\right] \text{diag}(\bm{\mu}(\omega))\cdot (\Q')^\top = \left[\begin{array}{c} \G \\ \bm{1}^\top\end{array}\right] \text{diag}(\bm{\mu}(\omega))\cdot (\Q^*)^\top.
\end{align*}  
The $l$-th row of the equation imply that for any $s_2\in \mathcal{S}_2$, 
\begin{align*}
	\sum_\omega \G[l,\omega] (\mu(\omega) \mu'(s_2|\omega)) = \sum_\omega\G[l,\omega](\mu(\omega) \mu^*(s_2|\omega)), 
\end{align*}
where $\G[l,\omega]$ represents the element in row $l$ and column $\omega$ of $\G$.
This is equivalent to 
\begin{align}\label{eqn:equi_Q}
	\sum_\omega \G[l,\omega] \mu'(\omega| s_2) \mu'(s_2) = \sum_\omega\G[l,\omega] \mu^*(\omega| s_2) \mu^*(s_2), \ \forall s_2 \in \mathcal{S}_2.
\end{align}
The last row means that for all $s_2 \in \mathcal{S}_2$,
\begin{align*}
	\sum_\omega \mu(\omega)\cdot \mu'(s_2|\omega) = \sum_\omega \mu(\omega)\mu^*(s_2|\omega),
\end{align*}
which is equivalent to 
\begin{align}\label{eqn:equi_mu}
	\mu'(s_2) = \mu^*(s_2), \ \forall s_2 \in \mathcal{S}_2.
\end{align}
Combining \eqref{eqn:equi_Q} and \eqref{eqn:equi_mu}, we get 
\begin{align}\label{eqn:equi_pos}
	\sum_\omega \G[l,\omega] \mu'(\omega| s_2)  = \sum_\omega\G[l,\omega] \mu^*(\omega| s_2), \ \forall s_2 \in \mathcal{S}_2.
\end{align}
This means that $\mu^*(\omega|s_2)$ and $\mu'(\omega|s_2)$ lead to the same report $\G\cdot \bm{\mu}^*(\omega|s_2) = \G \cdot \bm{\mu}'(\omega|s_2)$ for any $s_2\in \mathcal{S}_2$, which completes our proof of condition (c).

Finally, it is easy to show that (a) will also be satisfied. Equations \eqref{eqn:equi_mu} and \eqref{eqn:equi_pos} together imply that the marginal distribution of $\r_2$ remains the same, i.e.,
$$
\mu'(\r_2) = \mu^*(\r_2).
$$

\subsection{Proof of Theorem~\ref{thm:posterior_nec}} \label{app:pos_nec}
First because of conditional independence we have $\mu(\s_{-i}|s_i) = \sum_{\omega\in\Omega} \mu(\omega|s_i) \mu(\s_{-i}|\omega)$, or equivalently,
$$
\mu(\s_{-i}|s_i) = \P_i^\mu \cdot \mu(\omega|s_i).
$$
If there exists $\mu$ and $i$ such that rank$(\P^\mu_i)<|\Omega|$, then the columns of $\P^\mu_i$ are linearly dependent. So there exists a non-zero vector $\mathbf{a}$ such that $$\P^\mu_i \cdot \mathbf{a}= 0.$$ Since the set of possible likelihood function $M_i$ contains a ball with radius $\varepsilon$,\footnote{Here $M_i$ contains a ball means that there is a ball $\mathcal{B}\subseteq \mathbb{R}^{|\mathcal{S}_i|\times|\Omega|}$ such that the intersection of the ball and the space of valid conditional distributions is still in $M_i$, i.e.,  $\mathcal{B} \cap (\Delta\mathcal{S}_i)^{|\Omega|} \subseteq M_i$. } there must exists a small enough $\delta$ such that $\bm{\mu}(\omega|s_i) + \delta \mathbf{a}$ is a possible prediction of the state. More specifically, there exists $\tilde{\mu}\in M$ with the same prior $\tilde{\mu}(\omega)= \mu(\omega)$ and 
$$
\tilde{\mu}(\omega|s_i) = \mu(\omega|s_i) + \delta \mathbf{a}. 
$$
In addition, by our definition of conditional independent knowledge, $\tilde{\mu}$ can further have the same likelihood function for other agents $\tilde{\mu}(\s_{-i}|\omega) = \mu(\s_{-i}|\omega)$, which leads to 
$$
\tilde{\mu}(\s_{-i}|s_i) = \P_i^\mu \cdot \tilde{\mu}(\omega|s_i) = \P_i^\mu \left( \mu(\omega|s_i) + \delta \alpha\right) =  \P_i^\mu \cdot \mu(\omega|s_i) = \mu(\s_{-i}|s_i).
$$ 
This is equivalent to 
$$
\tilde{\mu}(\r_{-i}|s_i) = \mu(\r_{-i}|s_i)
$$
because $\tilde{\mu}(\omega|s_j)$ remains the same for all $j\neq i$. So agent $i$'s posterior about others' reports is the same but the truthful report becomes different $\tilde{\mu}(\omega|s_i) \neq \mu(\omega|s_i)$. This violates the necessary condition for single-tasks problems, i.e., the robust stochastic relevance condition Proposition~\ref{thm:necessary_1}, which requires
$$
Q_i(\tilde{\mu}(\omega|s_i)) \cap Q_i(\mu(\omega|s_i)) = \emptyset.
$$
It also violates the necessary condition for multi-task problems Proposition~\ref{thm:posterior_nec} since we have $\mu(\r_{-i})= \tilde{\mu}(\r_{-i})$.

\end{document}